\documentclass[sigconf,10pt]{acmart}

\usepackage{enumitem}
\usepackage{graphicx}
\usepackage{balance}  
\usepackage{hyperref}
\usepackage{amsmath}
\usepackage[linesnumbered,ruled,vlined]{algorithm2e}
\usepackage{algpseudocode}
\usepackage{latexsym}
\usepackage{multirow}
\usepackage{multicol}
\usepackage{color}
\usepackage{dashrule}
\usepackage{extarrows}
\usepackage{dsfont}

\usepackage{lipsum}                     
\usepackage[bottom]{footmisc}
\usepackage{booktabs} 
\usepackage{float}
\usepackage{todonotes}
\usepackage{gensymb}
\usepackage{url}

\DeclareMathOperator*{\argmin}{argmin}

\def\defterm#1{\textbf{#1}}
\def\set#1{\bs{#1}}
\def\bs#1{\boldsymbol{#1}}

\newcommand{\Xvariable}{X}
\newcommand{\Yvariable}{Y}
\newcommand{\Zvariable}{Z}

\newcommand{\dset}{D}

\newcommand{\test}{T}

\newcommand{\statistic}[1]{\phi(#1)}
\newcommand{\cstat}{Q}

\newcommand{\xvalue}{x}
\newcommand{\yvalue}{y}
\newcommand{\zvalue}{z}

\definecolor{nav}{RGB}{0,0,128}

\makeatletter
\newcommand*{\indep}{%
	\mathbin{%
		\mathpalette{\@indep}{}%
	}%
}
\newcommand*{\nindep}{%
	\mathbin{
		\mathpalette{\@indep}{/}
	}
}
\newcommand*{\@indep}[2]{%
	\sbox0{$#1\perp\m@th$}%
	\sbox2{$#1=$}%
	\sbox4{$#1\vcenter{}$}%
	\rlap{\copy0}%
	\dimen@=\dimexpr\ht2-\ht4-.2pt\relax
	\kern\dimen@
	\ifx\\#2\\%
	\else
	\hbox to \wd2{\hss$#1#2\m@th$\hss}%
	\kern-\wd2 %
	\fi
	\kern\dimen@
	\copy0 %
}
\makeatother

\newcommand{\SC}{\textsc{SC}\xspace}
\newcommand{\system}{\textsc{CODED}\xspace}
\newcommand{\SCs}{\textsc{SC}s\xspace}
\newcommand{\IC}{\mathcal{I}}
\newcommand{\DC}{\mathcal{D}}
\newcommand{\prob}{\mathbb{P}}
\newcommand{\varset}{\set{V}}
\newcommand{\na}[1]{{{\textcolor{nav}{\{\{Nathan: \bf #1\}\}}}\xspace}}

\newcommand{\oliver}[1]{{{\textcolor{blue}{\{\{Oliver: \bf #1\}\}}}\xspace}}

\newcommand{\graphoid}{\mathcal{G}}
\newcommand{\niepert}{\mathcal{A}}

\newtheorem{Definition}{Definition}

\newtheorem{Lemma}{Lemma}

\newtheorem{Theorem}{Theorem}

\newtheorem{Proposition}{Proposition}
\newtheorem{Conjecture}{Conjecture}
\newtheorem{Problem}{Problem}
\newtheorem{example}{Example}

\renewcommand\footnotetextcopyrightpermission[1]{} 
\acmDOI{xxx}

\acmISBN{xxx}

\acmConference[SIGMOD'19]{ACM SIGMOD Conference}{July 2019}{Amsterdam, The Netherlands}
\acmYear{1997}
\copyrightyear{2019, Amsterdam, The Netherlands}

\begin{document}


\title{Detecting Data Errors with Statistical Constraints}



\author{Jing Nathan Yan}
\authornote{Work done while visiting Simon Fraser University}
\affiliation{The University of Hong Kong}
\email{jyan@cs.hku.hk}

\author{Oliver Schulte}
\affiliation{Simon Fraser University}
\email{oschulte@sfu.ca}

\author{Jiannan Wang}
\affiliation{Simon Fraser University}
\email{jnwang@sfu.ca}

\author{Reynold Cheng}
\affiliation{The University of Hong Kong}
\email{ckcheng@cs.hku.hk}

\begin{abstract} 
A powerful approach to detecting erroneous data is to check which potentially dirty data records are incompatible with a user's domain knowledge. Previous approaches
allow the user to specify domain knowledge in the form of logical constraints (e.g., functional dependency and denial constraints). We extend the constraint-based approach by introducing a novel class of statistical constraints (\SCs). An \SC treats each column as a random variable, and enforces an independence or dependence relationship between two (or a few) random variables. Statistical constraints are expressive, allowing the user to specify a wide range of domain knowledge, beyond traditional integrity constraints. Furthermore, they work harmoniously with downstream statistical modeling. We develop \system, an S\underline{C}-\underline{O}riented \underline{D}ata \underline{E}rror \underline{D}etection system that supports three key tasks: (1) Checking whether an \SC is violated or not on a given dataset, (2) Identify the top-$k$ records that contribute the most to the violation of an \SC, and (3) Checking whether a set of input \SCs have conflicts or not. We 
present effective solutions for each task. Experiments on synthetic and real-world data illustrate how \SCs apply to error detection, and provide evidence that \system performs better than state-of-the-art approaches. 

\end{abstract}
\maketitle
\section{Introduction}\label{sec:intro}

Error detection, the discovery of erroneous values from a database, has been a long-standing problem~\cite{abedjan2016detecting}.
Data errors can yield wrong decision making and biased machine-learning models. In 2016, IBM estimated that poor data quality costs the U.S. economy around \$3 trillion per year~\cite{dirty-data-cost}. As more companies 
assign data a central place in their business, 
the impact of data errors continues to grow.





Constraint-based error detection is one of the most widely used approaches~\cite{DBLP:journals/ftdb/IlyasC15}. 
A user represents domain knowledge by specifying a constraint that describes what the data should look like; the system detects which parts of the data violate the constraint. Take functional dependencies (FDs) as an example for a hospital table 
with \textsf{\small Name}, \textsf{\small Address}, \textsf{\small Zipcode}, and \textsf{\small City} attributes. Suppose a user specifies a FD:  $\textsf{\small Zipcode}  \rightarrow \textsf{\small City}$, which means that if two records have the same $\textsf{\small Zipcode}$ value, then they must have the same value in the \textsf{\small City} Column. By comparing every pair of records in the table, if two records have the same entry in $\textsf{\small Zipcode}$ but different entries in $\textsf{\small City}$, an error that violates the constraint is detected. 




\sloppy
Existing approaches typically use integrity constraints (ICs), such as functional dependencies~(FDs)~\cite{bohannon2005cost}, conditional functional dependencies~(CFDs)~\cite{bohannon2007conditional}, and denial constraints~(DCs)~\cite{chu2013discovering}, to express user's domain knowledge.  They model each row as a real-world entity and compare two (or a few) entities to detect violations. In the above example, each row represents a hospital, 
and an FD violation is detected by comparing two hospitals. 

\fussy




\begin{table}[t]
\centering \small\fontsize{8pt}{8pt} \selectfont
\caption{A comparison between ICs and SCs.\label{long}} \vspace{-1em}
{\renewcommand{\arraystretch}{1.5}%
\begin{tabular}{ |@{\;}c@{\;}|@{\;}l@{\;}|@{\;}l@{\;}| } \hline 
  & {\bf Data Model} & {\bf Violation Definition}\\ \hline
 \hline 
 {\bf IC} & {\fontfamily{ptm}\selectfont Row as an entity} & {\fontfamily{ptm}\selectfont Between two (a few) entities}\\  \hline
 {\bf SC} & {\fontfamily{ptm}\selectfont Column as a random variable} & {\fontfamily{ptm}\selectfont Between two (a few) random variables} \\ 
 \hline
\end{tabular}}\vspace{-1.5em}
\end{table}

In this paper, we introduce a novel class of constraints, named \emph{Statistical Constraints} (\SCs). 
\SCs model each column as a {\em random variable}, and compare two (a few) random variables to detect violations. An \SC specifies an \emph{independence} or \emph{dependence} relationships between random variables: 
\emph{X and Y are dependent (or independent) given Z}, where X, Y, and Z is either a single random variable or a set of variables. Intuitively, this \SC means that knowing X will (or will not) reveal any information about Y given Z. \SCs have been studied extensively in machine learning, statistics, and AI~\cite{Niepert2013,pearl2}. The main 
purpose of this paper is to i) introduce \SCs and associated techniques to data cleaning and ii) study how to apply \SCs to error detection. 

In Figure~\ref{fig:car-error}(a), suppose that a user specifies $\SC_1 = $ \emph{``\textsf{RID} and \textsf{Price} are independent given \textsf{Model}''}, which means that for each car model, knowing the random row id of a car should not reveal any information about the car's price. However, on this example dataset, for  \textsf{Model = ``BMW X1''}, \textsf{RID} is sorted by \textsf{Price}. That is, if we know a BMW X1 car has a smaller \textsf{RID}, it will have a lower price. \textsf{RID} and \textsf{Price} are highly dependent 
given \textsf{Model = ``BMW X1''}, which violates $\SC_1$. This kind of sorting error may cause issues in machine learning. For example, a sorting error was found in KDD-Cup 2008, which dealt with cancer detection from mammography data. A team found this error in the training set and utilized the dependence relationship between the ``Patient ID'' and the class label to win the competition~\cite{rosset2010medical}. In reality, however, the ``Patient ID'' should not be used to predict whether a patient has cancer or not.


In Figure~\ref{fig:car-error}(b), suppose that a user specifies $\SC_2 = $\emph{``\textsf{Model} and \textsf{Color} are independent''}, which means that knowing the model of a car should not reveal any information about the car's color. However, on this example dataset, each 
\textsf{``Toyota Prius''} is assigned the color ``White'' (highlighted in Figure~\ref{fig:car-error}(b)). Therefore, \textsf{Model} and \textsf{Color} are highly dependent 
which violates $\SC_2$. This kind of error is common in practice since real-world data often has missing values (e.g., the colors of all Toyota cars are missing) and a data collector may impute the missing values with a default value (e.g., ``White''). If a data user is not aware of the imputation process, data analysis may lead to faulty conclusions.

{\em Advantages} of \SCs include the following. (1) {\em Interpretability.} (In)dependencies are easily interpreted by the user as causal or statistical ir(relevance) among attributes.
Although 
\SCs have a probabilistic semantics, 
specifying an (in)dependence constraint does not require the user to specify or even consider numeric values. Many tools for visualizing dependencies are available~\cite{Hall2009}. The large field of graphical models is based on the insight that (in)dependence constraints can be represented in terms of purely qualitative graphical relations among variables~\cite{pearl2}. (2) {\em Detectability.} The field of statistical hypothesis testing has developed many methods
for deciding whether a given data set violates an 
\SC~\cite{wasserman2013all}. These methods provide parameters for controlling false positive and false negative error rates. (3) {\em Expressive power.}  \SCs are fundamentally different from ICs and complement them by allowing a user to express additional knowledge about a domain. 
Specifically, \SCs allow a user to express not only that some attributes are relevant to others but also {\em ir}relevance relationships. 
For example, consider a Car table with \textsf{Model}, \textsf{Fuel Efficiency}, \textsf{Price}, and \textsf{Color} columns (see Figure~\ref{fig:running-example}). A user can directly specify an dependence \SC: \emph{\textsf{Fuel-Efficiency}  and \textsf{Price} are dependent}, or an independence \SC: \emph{\textsf{Model}  and \textsf{Color} are independent}, but is not clear how to use ICs to specify these relationships. Section~\ref{sec:scvsic} discusses the expressive power of \SCs formally.

\SCs are expected to hold only approximately, 
which allows for exceptions and increases their applicability. Thus they provide a strong foundation for defining approximately satisfied constraints.
Our experiments compare \SCs with approximate FDs. Because of their statistical basis, \SCs work harmoniously with downstream statistical modeling
(e.g., building a regression model). \SC-based error detection attempts to highlight unreasonable relationships among columns; statistical modeling seeks to learn correct relationships among columns. Suppose that a statistical model learns from dirty data that \emph{\textsf{\small Fuel-Efficiency} and \textsf{\small Price} are independent}. If a user suspects that there are errors that cause the model to learn such an unreasonable relationship, she can 
specify an \SC ``\emph{\textsf{\small Fuel-Efficiency} and \textsf{\small Price} are dependent}'' to detect the errors.  

\begin{figure}[t]
   \centering
   \includegraphics[width=1\linewidth]{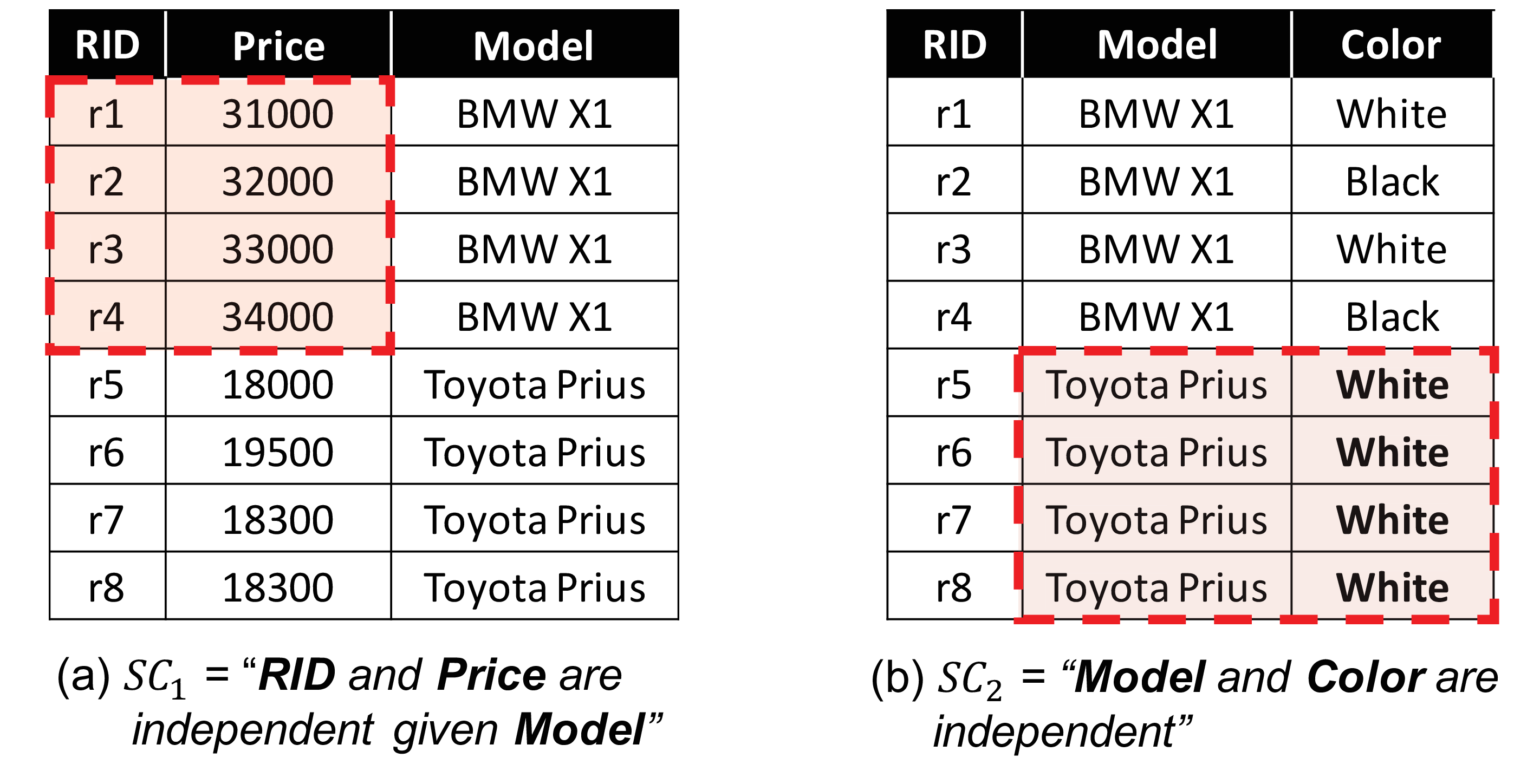} \vspace{-2em}
   \caption{Examples of data errors detected by \SCs}\label{fig:car-error}
\end{figure}


 
 \vspace{.25em}

{\em Research Challenges.} We revisit a number of important research problems for constraint-based error detection in the context of \SCs. 

\vspace{.25em}


{\noindent \bf \SC Violation Detection.} Given a (dirty) dataset and an \SC, we need to check whether the dataset violates the \SC.  
The main novel challenge is that an \SC does not hold absolutely on a dataset, but only to a degree. 
We leverage Hypothesis Testing and  well-established statistical metrics to quantify the degree to which an \SC is compatible with a dataset.



\vspace{.25em}

{\noindent \bf Error Drill Down.} To explain an \SC violation, we further study how to drill down into erroneous rows rather than just return erroneous columns. We propose two effective top-$k$ algorithms, aiming to return the top-$k$ most likely erroneous rows. For example, suppose a user specifies an \SC: \emph{\textsf{\small Model} and \textsf{\small Color} should be independent}, but the dataset violates the independence relationship. The user may want to manually examine a small number of rows from the data in order to reason about the violation. Given a user-specified threshold $k$ (e.g., 100), the top-$k$ algorithms can identify the top-100 rows which have the biggest impact on the violation. 

\vspace{.25em}

{\noindent \bf Consistency Checking.} Given a set of \SCs, the consistency-checking problem is to determine whether they have conflicts. This problem has been studied in the AI literature~\cite{pearl2,studeny1990conditional}, but it is still not known whether it is decidable~\cite{Niepert2013}. We prove that the consistency-checking algorithm based on Graphoid Axioms~\cite{pearl2} is pseudo-polynomial in $n$ and $\ell$, where $n$ is the number of input \SCs and $\ell$ is the largest number of variables that occurs in any of the input \SCs. The proof gives a theoretical insight 
into the efficiency of the consistency-checking algorithm.

We develop \system, a S\underline{C}-\underline{O}riented \underline{D}ata \underline{E}rror \underline{D}etection system, which implements \SC violation detection, error drill down, and consistency checking. 
Extensive experiments on synthetic and real-world datasets demonstrate the advantages of \system over state-of-the-art baselines. To summarize, our contributions are:

%
 \vspace{-.25em}
%
%
%
\begin{itemize}[leftmargin=*]
    \item We study the use of Statistical Constraints (\SCs) in error detection, and identify situations in which \SCs complements ICs. 
    \item We study how to check whether a dataset violates an \SC, 
    and apply Hypothesis Testing to quantify the degree to which an \SC holds or fails.
    \item We study how to drill down into erroneous rows w.r.t. an \SC violation, and propose two top-$k$ algorithms to solve it. 
    \item We study the consistency-checking problem and provide a theoretical justification for the efficiency of the algorithm based on Graphoid Axioms.
   \item We conduct extensive experiments on real-world datasets. The results demonstrate the superiority of our approaches over the state-of-the-art baseline approaches. 
\end{itemize}

\vspace{-.25em}

The remainder of this paper is organized as follows. We review related work in Section~\ref{sec:rw}, and formally define \SCs in Section~\ref{sec:sc}. Section~\ref{sec:sys-arch} presents the \system system architecture. We study the \SC violation detection problem in Section~\ref{sec:detection}, the error-drill-down problem in Section~\ref{sec:drill-down}, and the consistency-checking problem in Section~\ref{cck}.  Section~\ref{sec:exp} reports our experimental findings. We conclude in Section~\ref{con}.






\section{Related Work} \label{sec:rw}


\begin{figure}[t]
   \centering
   \includegraphics[width=0.75\linewidth]{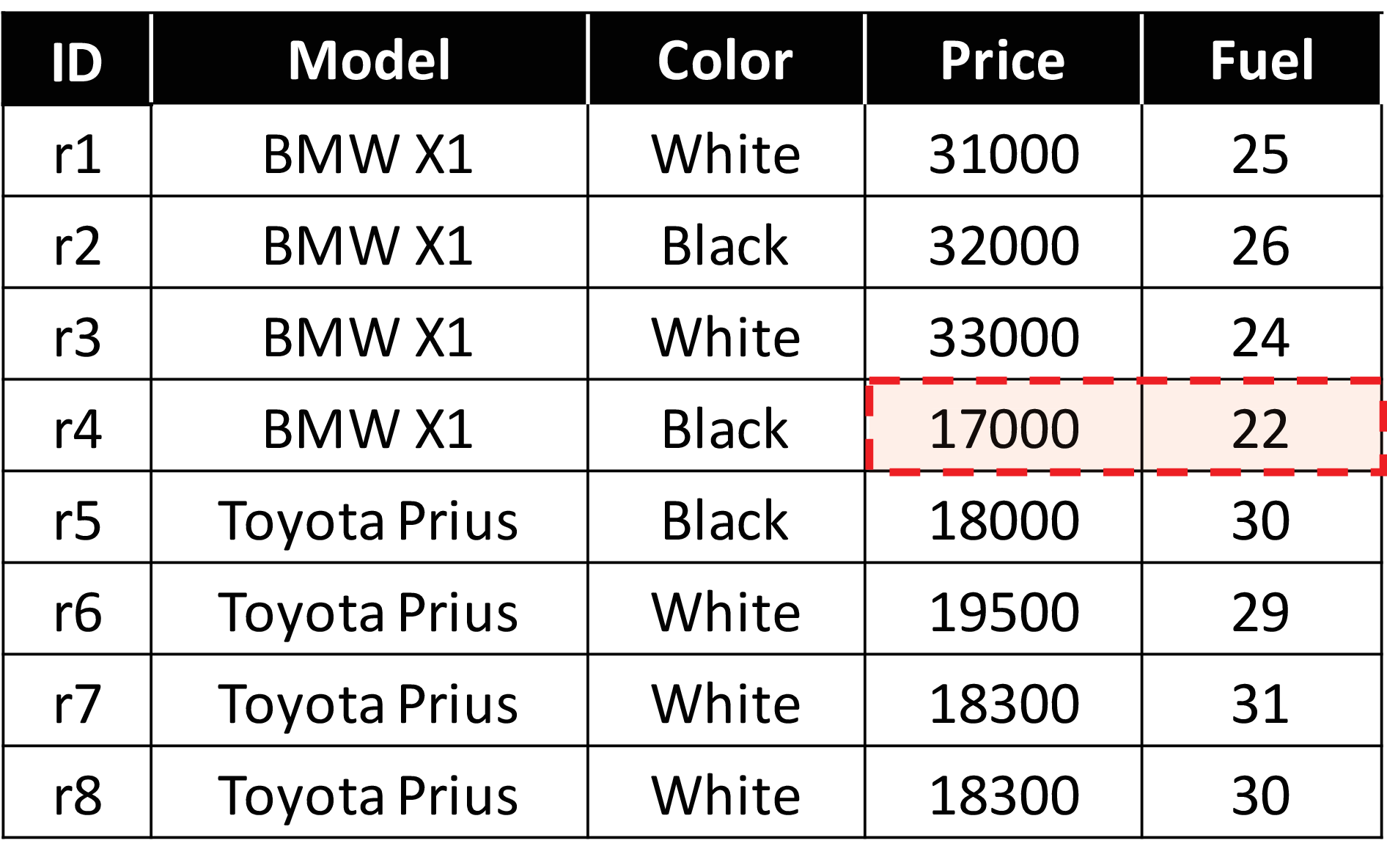} \vspace{-1em} 
   \begin{align*}
       & \SC_1 = \textsf{Model} \indep \textsf{Color} & & \SC_2 = \textsf{Fuel} \indep \textsf{Model}\\
       & \SC_3 = \textsf{Price} \indep \textsf{Fuel} \mid \textsf{Model}  &  & \SC_4 = \textsf{Model} \nindep \textsf{Fuel}~|~\textsf{Price}
   \end{align*}\vspace{-2em}
   \caption{ A Running Example ($\SC_2$ could be removed. The highlighted cells are the data errors detected by \system w.r.t. $\SC_4$.). }\label{fig:running-example}\vspace{-1em}
\end{figure}

{\noindent \bf Constraint-based Error Detection.} A recent survey paper classifies error-detection approaches into four categories: constraint-based, pattern-based, outlier detection, and de-duplication~\cite{abedjan2016detecting}. Our paper belongs to the first category, where a user specifies a constraint, and the system detects erroneous values that violate the constraint. Various forms of ICs are proposed for error detection~\cite{DBLP:journals/ftdb/IlyasC15}. 
In this paper, we introduce a new class of statistical constraints (\SCs) and study how to apply \SCs to error detection. 
\SCs are more expressive than ICs in certain situations, and work harmoniously with downstream statistical modeling. A more detailed comparison between ICs and \SCs can be found in Section~\ref{sec:scvsic}.  

\SCs have been studied in the statistic and AI literature~\cite{dawid1979conditional,Pearl2000}. Existing studies assume that data is clean and explore how to infer \SCs from the data. The derived \SCs can be used for statistical modeling and casual inference. Ilyas et al. show that \SCs (involving pairs of columns) are effective in improving query optimization~\cite{ilyas2004cords}. Salimi et al. leverage (conditional) independence relationships among attributes to resolve bias in OLAP queries~\cite{DBLP:conf/sigmod/SalimiGS18}. However, none of these works has studied how to apply \SCs to detect data errors. 





\vspace{.5em}

{\noindent \bf Other Error Detection Approaches.} There are other kinds of error detection approaches. Some works study how to leverage regular expressions or external resources to detect wrong data types or inconsistent data formats in a single column~\cite{DBLP:conf/vldb/RamanH01,yeye18auto,yeye18sync}. Unlike them, we detect the violation of the statistical relationship among multiple columns. Outlier detection typically leverages the data distribution of a single column to detect errors (e.g., any datapoint that is more than 3 standard deviation is an outlier)~\cite{hellerstein2008quantitative}. Some works also study how to leverage the relationships between multiple columns to detect outliers~\cite{das2008anomaly,riahi2015model,mariet2016outlier,DBLP:conf/kdd/DasS07}. However, none of them allows a user to specify a set of SCs explicitly and then guides the user to detect the errors based on the SCs. KATARA~\cite{chu2015katara} relies on the semantic relationships between columns (e.g., column A is the \emph{capital} of column B) to detect errors. While this approach is powerful, it requires the availability of knowledge bases and crowdsourcing. 

\vspace{.25em}


{\noindent \bf Data Repairing.} Data cleaning has two important tasks: error detection and data repairing. Our paper focuses  on the former task. Next, we review some related work about data repairing~\cite{chu2016data}. Data repairing studies the problem of correcting erroneous values. We classify existing works into two categories. 

(1) The first category of works combines data repairing with error detection. Given a set of input constraints, they study how to find the minimum change to the data to satisfy the input constraints. Most works assume that the input constraints are ICs~\cite{chu2013holistic,DBLP:conf/sigmod/BohannonFFR05,DBLP:conf/icdt/KolahiL09}. Prokoshyna et al. consider both ICs and statistical distortion~\cite{DBLP:journals/pvldb/ProkoshynaSCMS15}. Note that statistical distortion~\cite{dasu2012statistical} is a different notion from \SCs. It is defined as the difference between two data distributions (the original data and the repaired data).

(2) The second category treats error detection as a black box and focus on data repairing~\cite{mayfield2010eracer,yakout2013don,DBLP:journals/pvldb/RekatsinasCIR17}. ERACER~\cite{mayfield2010eracer} takes as input a dataset with missing attribute values, and utilizes belief propagation and relational dependency networks to infer the missing values.  SCAREd~\cite{yakout2013don} takes as input a dataset with a subset of rows identified as dirty and leverages maximal likelihood for data repairing. HoloClean~\cite{DBLP:journals/pvldb/RekatsinasCIR17} assumes that erroneous values have been marked by users and then automatically generates a probabilistic program to correct the erroneous values. These approaches are orthogonal to our work. For example, a user can first run our approach to detect erroneous values and then use one of the solutions to repair those values.

\vspace{.5em}

{\noindent \bf Analysis-aware Data Cleaning.} There are some studies on how to clean data for downstream statistical and SQL analysis. SampleClean is a framework that enables fast and reliable query processing over dirty data~\cite{krishnan2015sampleclean}. ActiveClean provides a progressive data-cleaning framework for statistical modeling~\cite{krishnan2016activeclean}. QOCO~\cite{DBLP:conf/sigmod/BergmanMNT15} uses an oracle to clean data that have largest impact on SQL queries.

\vspace{.5em}

{\noindent \bf Error Explanation.} Error explanation, which aims to provide intuitive explanations to the user about the errors, is another hot topic in data cleaning~\cite{wu2013scorpion,wang2017qfix,DBLP:conf/sigmod/RoyS14,DBLP:conf/sigmod/WangDM15}. For example, Scorpion~\cite{wu2013scorpion} uses sensitivity analysis to identify a set of data records that contribute the most to an outlier. Our error-drill-down component adopts a similar idea, but the difference is that we identify the top-$k$ data records that most influence the violation of an \SC rather than an outlier. 

\vspace{-.5em}
\section{Statistical Constraints}\label{sec:sc}
In this section, we first present a formal definition of \SCs, then discuss how to facilitate users to discover \SCs from data, finally discuss several advantages of using \SCs in practice. 

\subsection{Definitions}\label{scd}

A \emph{variable} $\Xvariable$ is an attribute or feature that can be assigned a value $\xvalue$ from a fixed domain; we write $\Xvariable = \xvalue$ to denote an assignment of a value to a variable. We use boldface vector notation for finite sets of objects. So for example $\set{\Xvariable} = \set{\xvalue} \equiv (\Xvariable_1 = \xvalue_1,\Xvariable_2 = \xvalue_2, \ldots, \Xvariable_n = \xvalue_n)$ denotes the \emph{joint assignment} where variable $\Xvariable_i$ is assigned value $\xvalue_i$, for each $i=1,\ldots,n$. In relational terms, a variable corresponds to an attribute or column, and a joint assignment to a tuple or row. 



A \emph{random variable} requires a distribution $P(\Xvariable = \xvalue)$ that assigns a probability to each domain value in the domain of $\Xvariable$. A \emph{joint distribution} $P(\set{\Xvariable} = \set{\xvalue})$ assigns a probability to each joint assignment. %
%
%
Given a joint distribution for a set of variables $\set{\Xvariable}$, the \emph{marginal distribution} over a subset $\set{\Yvariable}$ is defined by $P(\set{\Yvariable}=\set{\yvalue}) \equiv \sum_{\set{\zvalue}} P(\set{\Yvariable}=\set{\yvalue},\set{\Zvariable}=\set{\zvalue}).$ Here $\set{\Zvariable} = \set{\Xvariable}- \set{\Yvariable}$ contains the set of variables in $\set{\Xvariable}$ but not in $\set{\Yvariable}$, and the comma notation $\set{\Yvariable}=\set{\yvalue},\set{\Zvariable}=\set{\zvalue}$ denotes the conjunction of two joint assignments. 
The \emph{conditional probability} of an assignment $\set{\Xvariable} = \set{\xvalue}$ given another assignment $P(\set{\Yvariable})=\set{\yvalue}$ is defined as
$P(\set{\Xvariable} = \set{\xvalue}|\set{\Yvariable}=\set{\yvalue}) \equiv P(\set{\Xvariable} = \set{\xvalue},\set{\Yvariable}=\set{\yvalue})/P(\set{\Yvariable}=\set{\yvalue}).$

A key notion of this paper is the concept of {\em conditional independence among sets of variables.} Intuitively, a set of variables $\set{\Xvariable}$ is independent of another set $\set{\Yvariable}$ given a third conditioning set $\set{\Zvariable}$ if knowing the values of the variables in $\set{\Yvariable}$ adds no information about the values of the variables in $\set{\Xvariable}$, beyond what can be inferred from the values in the set $\set{\Zvariable}$. Formally, for three disjoint sets $\set{\Xvariable}, \set{\Yvariable}, \set{\Zvariable}$ and assume  $P(\set{\Zvariable} = \set{\zvalue}) > 0$ for all values~$\set{\zvalue}$  we define %
\begin{eqnarray*}
\set{\Xvariable} \indep \set{\Yvariable} | \set{\Zvariable} & \equiv & \\ P(\set{\Xvariable}=\xvalue,\set{\Yvariable}=\yvalue|\set{\Zvariable} = \zvalue) & = \\ P(\set{\Xvariable}=\xvalue|\set{\Zvariable}=\zvalue) \times P(\set{\Yvariable}=\yvalue|\set{\Zvariable} = \zvalue) & \mbox{for all } \xvalue, \yvalue, \zvalue.
\end{eqnarray*}








We call $\set{\Xvariable} \indep \set{\Yvariable} | \set{\Zvariable}$ a \emph{(conditional) independence statement}. A \emph{(conditional) dependence statement}, written $\set{\Xvariable} \nindep \set{\Yvariable} | \set{\Zvariable}$, holds if for {\em some} values $\set{\xvalue},\set{\yvalue},\set{\zvalue}$, we have $P(\set{\Yvariable} = \set{\yvalue},\set{\Zvariable} = \zvalue)>0$ and $P(\set{\Xvariable}=\xvalue|\set{\Yvariable}=\yvalue,\set{\Zvariable} = \zvalue) \neq P(\set{\Xvariable}=\xvalue|\set{\Zvariable} = \zvalue).$ A set of Statistical Constraints (\SCs) comprises a set of independence statements and dependence statements.

 

\begin{Definition}[Statistical Constraints] Fix a set of variables $\varset = \{V_1,\ldots,V_n\}$. A finite set of statistical constraints $\Sigma = \IC \cup \DC$ comprises
	\begin{enumerate}
		\item a finite set of independence \SCs, $\IC = \{\phi_1,\ldots,\phi_p\}$, where each $\phi_i$ is of the form $\set{\Xvariable} \indep \set{\Yvariable} | \set{\Zvariable}$, and
		\item a finite set of dependence \SCs, $\DC = \{\lambda_1,\ldots,\lambda_q\}$, where each $\lambda_i$ is of the form $\set{\Xvariable} \nindep \set{\Yvariable} | \set{\Zvariable}$ 
	\end{enumerate}
\end{Definition}


Figure~\ref{fig:running-example} shows four \SCs: $\Sigma = \{\SC_1, \SC_2, \SC_3, \SC_4\}$, where $\IC = \{\SC_1, \SC_2, \SC_3\}$ contains three independence \SCs and $\DC = \{\SC_4\}$ contains one dependence \SC. For instance, $\SC_3 = \textsf{Price} \indep \textsf{Fuel} \mid \textsf{Model}$ means that $\textsf{Price}$ is independent of $\textsf{Fuel}$ given $\textsf{Model}$; $\SC_4 = \textsf{Price} \nindep \textsf{Fuel}$ means that $\textsf{Price}$ is dependent on $\textsf{Fuel}$. It is useful to distinguish the different types of \SC statements shown in Table~\ref{table:ci-types}.



\begin{table}[t]
\caption{Different types of independence \SCs ~\cite{Niepert2013}. 
}
\vspace{-1em}
\begin{center}
\resizebox{0.5\textwidth}{!}{
\begin{tabular}{|c|c|c|}
\hline
    \textbf{Type} & \textbf{Definition} & \textbf{Example}\\\hline
\hline
\textbf{Elementary} & 
\begin{tabular}{c}
$\Xvariable \indep \Yvariable|\set{\Zvariable}$ \\
   \mbox{ where } $\Xvariable,\Yvariable$ are single variables
\end{tabular} & $\textsf{Price} \indep \textsf{Fuel} \mid \textsf{Model} $ \\\hline
\textbf{Marginal} & $\Xvariable \indep \Yvariable$ & $\textsf{Price} \indep \textsf{Fuel}$ \\\hline
\textbf{Saturated}& 
\begin{tabular}{c}
$\Xvariable \indep \Yvariable|\set{\Zvariable}$ \\
 \mbox{ where } $\Xvariable,\Yvariable,\set{\Zvariable}$ comprise all variables
\end{tabular}
& $\textsf{Price} \indep \textsf{Fuel} \mid \textsf{Model},\textsf{Color} $ \\\hline
\end{tabular}
}
\end{center}
\label{table:ci-types} \vspace{-1em}
\end{table}%

\subsection{Statistical vs. Integrity Constraints} \label{sec:scvsic} 
Integrity constraints enforce deterministic Boolean conditions on {\em sets} (relations), whereas statistical constraints impose (in)equalities on {\em distributions}, defined by the {\em cardinalities} of the relevant sets or relations. Both logical and probabilistic dependencies represent {\em inferential relevance}: reasoning from values in one set of columns to values in another. The next proposition compares the logical strength of integrity and statistical constraints that assert dependencies. We say that one constraint implies another if any data table $\dset$ that satisfies the former also satisfies the latter. A table $\dset$ satisfies a probabilistic constraint if the data distribution $P_{\dset}$ does. 
%
We write  $\Rightarrow$ to denote that one constraint implies another. 

\begin{Proposition} Let $\rightarrow$ denote functional dependence and $\rightarrow \rightarrow$ denote multi-valued dependence. 
\begin{align}
   \set{X} \rightarrow \Yvariable \Rightarrow \set{X} \rightarrow \rightarrow \Yvariable \Rightarrow \set{X} \nindep \Yvariable \label{eq:dependence} \\
   \set{X} \rightarrow Y \Rightarrow (\set{X} \cup \Yvariable)^{C}  \indep Y | \set{X} \Rightarrow \set{X} \rightarrow\rightarrow \Yvariable \label{eq:independence}
\end{align}
where the first equation 
assumes that not every $\set{\xvalue}$-tuple is related to exactly the same set of $\yvalue$-values.
\end{Proposition}

\begin{proof}
Equation~\eqref{eq:dependence}: It is well-known that an FD implies an MVD~\cite{fagin1977multivalued}. Consider an MVD $\set{X} \rightarrow \rightarrow \Yvariable$. Since not every every $\set{\xvalue}$-tuple is related to exactly the same set of $\yvalue$-values, there exists $\set{\xvalue_1}, \set{\xvalue_2},\yvalue$ such that $\set{\xvalue_1}$ is related to $\yvalue$ but $\set{\xvalue_2}$ is not related to $\yvalue$. Therefore $P_{\dset}(\yvalue|\set{\xvalue_1})>0 = P_{\dset}(\yvalue|\set{\xvalue_2})$. Therefore $\set{\Xvariable} \nindep \set{\Yvariable} | \set{\Zvariable}$ as defined in Section~\ref{sec:sc}. 
Equation~\eqref{eq:independence}: Consider an FD  $\set{X} \rightarrow Y$. Let $\set{Z} \equiv (\set{X} \cup \Yvariable)^{C}$. The FD implies that for every $\set{\xvalue}$ tuple, there is a unique value $\yvalue \equiv f(\xvalue)$ value such that for every $\set{\zvalue}$-tuple, the tuple $\langle \set{\xvalue}, \set{\zvalue}\rangle$ is related to $f(\set{\xvalue})$. Therefore $P_{\dset}(f(\set{\xvalue})|\set{\xvalue}) = 1 = P_{\dset}(f(\set{\xvalue})|\set{\xvalue},\set{\zvalue})$. So the conditional distribution $P_{\dset}(\yvalue'|\set{\xvalue}\set{\zvalue})$ is the same for every $\set{\zvalue}$-tuple for the variables in $\set{X} \cup \Yvariable)^{C}$. Since this holds for each $\set{\xvalue}$-tuple,  we have $(\set{X} \cup \Yvariable)^{C}  \indep Y | \set{X}$.

For the second implication, we show its contrapositive: if there is no multi-valued dependency $\set{X} \rightarrow\rightarrow \Yvariable$, then $(\set{X} \cup \Yvariable)^{C}  \nindep Y | \set{X}$. If there is no multi-valued dependency, $\set{X} \rightarrow\rightarrow \Yvariable$ in the data table $\dset$, then there exist tuples $\set{\xvalue}, \set{\zvalue_1}, \set{\zvalue_2}$ such that the tuples $\langle \set{\xvalue}, \set{\zvalue_1}\rangle$ and $\langle \set{\xvalue}, \set{\zvalue_2}\rangle$ are related to different sets of $\Yvariable$-values. Without loss of generality, suppose that $\yvalue_1$ is related to $\langle \set{\xvalue}, \set{\zvalue_1}\rangle$ but not to $\langle \set{\xvalue}, \set{\zvalue_2}\rangle$. Then $P_{\dset}(\yvalue_1|\set{\xvalue}\set{\zvalue_1}) > 0 = P_{\dset}(\yvalue_1|\set{\xvalue}\set{\zvalue_2})$. So the conditional distribution $P_{\dset}(\yvalue_1|\set{\xvalue}\set{\zvalue})$ is not the same for all $\set{\zvalue}$, which implies that $(\set{X} \cup \Yvariable)^{C}  \nindep Y | \set{X}$. 
\end{proof}

The deterministic aspect of an FD corresponds to the saturated independence constraint~\eqref{eq:independence}, which says that the variables $\set{Z}$ determine a unique {\em conditional distribution} for $Y$ regardless of the other attributes. Since a distribution determines a unique support set---comprising the values with positive probability---the independence constraint $(\set{X} \cup \Yvariable)^{C}  \indep Y | \set{X}$ entails a {\em multi-valued dependency} between columns $X$ and $Y$~\cite{fagin1977multivalued}.


 To illustrate equation \eqref{eq:dependence} in our running example (Figure~\ref{fig:car-error}), suppose that all cars of the same model have the same price. This entails an FD $\textsf{Model} \rightarrow \textsf{Price}$. If some model offers a range of prices, the FD is violated, but the MVD $\textsf{Model} \rightarrow \rightarrow \textsf{Price}$ may hold if for each model, there is a unique set of prices offered. Unless this set of prices is the same for all models, a probabilistic dependence $\textsf{Model} \nindep \textsf{Price}$ holds in the dataset. In the case of an FD, we expect the probabilistic dependence to be very strong. 
 
 For equation~\eqref{eq:independence} , the FD $\textsf{Model} \rightarrow \textsf{Color}$ does not hold because different car models may come in different colors; still the independence 
 $\textsf{Model} \indep \textsf{Color}$ holds if each model offers the same distribution 
 of colors. The independence implies the multi-valued dependency because if the distribution of offered colors is the same, then so is the set of offered colors.

 Because they assert inferential {\em ir}relevance, marginal  
 
 \noindent {\em in}dependence constraints cannot be captured by FDs or MVDs. The same applies to more complex logical dependencies, such as Denial Constraints: they correspond to probabilistic marginal dependencies and saturated independencies, but not marginal independencies.
 
 As our example illustrates, \SCs can {\em express new restrictions} and hence new domain knowledge. We emphasize that our example also illustrates that \SCs as a class are not strictly more expressive than ICs: for instance an FD is more informative than the associated probabilistic dependence. 
 The additional expressive power of \SCs stems from two sources. 1) \SCs can be negated so the user can express both dependencies and (non-saturated) independencies. 2) The traditional application of \SCs is to sample data, where we expect \SCs to hold only approximately with exceptions. As will be discussed in Section~\ref{sec:detection}, this is accomplished by defining {\em degrees of dependence} (e.g. correlations in (0,1]). 
 Thus a user can specify a valid constraint even if she does not expect it to hold exactly. For example, she may not expect that for each model, the proportion of available colors is exactly the same. 

\subsection{Discovering Input \SCs}\label{sec:sc-discovery}

\begin{figure}[t]
   \centering
   \includegraphics[width=1\linewidth]{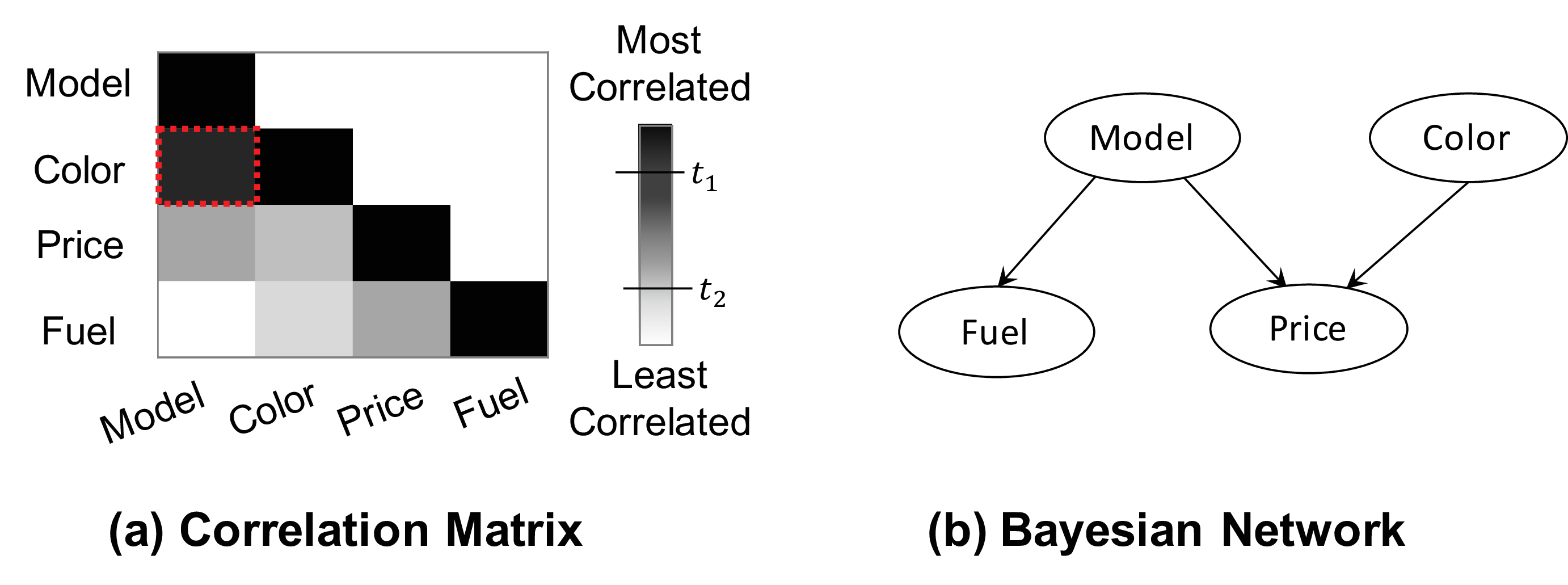} \vspace{-2em}
   \caption{Illustration of the use of correlation matrix (a), computed from the data, to facilitate the user to discover \SCs. The (dirty) data, show a false correlation between Model and Color. The user's domain knowledge allows her to record an independence \SC: $\textsf{Model} \indep \textsf{Color}$, also represented in (b). 
}\label{fig:sc-discovery}
\end{figure}

\system requires a set of \SCs as input. An expert may simply list known associations. A famous example from data mining is that in Tesco stores, purchases of diapers correlated with purchases of beer. In addition, 
the fields of artificial intelligence, graphical modelling, and machine learning provide many resources that can be leveraged to elicit \SCs from a user. We briefly describe some of the main ideas. 

\vspace{0.25em}

{\noindent \bf{Building a Model Structure}}. 
A traditional and effective approach to building expert systems with graphical models~\cite{Cowell2006} is to ask the domain expert to express his or her domain knowledge about statistical relationships by specifying the structure of a graphical model, as shown in Figure~\ref{fig:sc-discovery}(b). This structure is a graph whose nodes are random variables; it represents statistical constraints in terms of graphical concepts such as the existence of a path between two nodes.  
One reason why graphical representations are meaningful to domain experts is that they visualize {\em causal patterns}~\cite{Pearl2000,Spirtes2000,wasserman2013all,salimi2018hypdb}. For instance, Figure~\ref{fig:sc-discovery}(b) shows that the car model causally influences its fuel consumption.

\vspace{0.25em}

{\noindent \bf{Mixed-Initiative Data Analysis.}} A mixed-initiative approach uses automated ML methods to find (in)dependencies in the data. This is similar to work on automatic IC discovery~\cite{chiang2011unified,chiang2008discovering, wyss2001fastfds,lopes2000efficient,chu2013discovering}. Since the data may be dirty, the results of the automated methods need to be checked by the domain expert. If the expert inspects a suggested independence and rejects it based on domain knowledge, we can add the dependence that represents its negation as a valid constraint; similarly a rejected dependence statement provides an independence constraint. %
For example, a correlation matrix shows the pairwise Pearson correlations among a set of variables (see Figure~\ref{fig:sc-discovery}(a)). Computing a correlation matrix from dirty data may wrongly show a high observed correlation between $\textsf{Model}$ and $\textsf{Color}$. Suggesting this correlation to the user is likely to elicit the \SC $\textsf{Model} \indep \textsf{Color}$. 


Table~\ref{table:discover-scs} shows several common data analysis techniques for discovering \SCs. Details on these techniques can be found in standard textbooks on data mining~\cite{Hall2009}. {\em Marginal (in)dependencies} can be visualized in scatter plots (continuous variables), paired histograms (categorical variables), and boxplots (mixed variables). The strength of an association can be quantified as well, which facilitates an algorithmic search for associations based on a strength treshold~\cite{ilyas2004cords}.  Finding {\em elementary saturated (in)dependencies} 
is a key machine learning task known as {\em feature selection}~\cite{Hall2009}. If a feature $\Xvariable$ is selected as relevant to predicting the value of a target variable $\Yvariable$, this implies a saturated dependence $\Xvariable \nindep \Yvariable|(\Xvariable \cup \Yvariable)^C$. 

\begin{table}[t]\vspace{-2em}
\caption{Data analysis approaches to discovering \SCs.
} \vspace{-1em}
\begin{center}
\resizebox{0.5\textwidth}{!}{
\begin{tabular}{|c|c|c|}
\hline
& \multicolumn{2}{|c|}{\textbf{Data Analysis}} \\\hline
\SC \textbf{format} & 
\textbf{Visualization} 
& \textbf{Strength Metric} \\\hline \hline
\textbf{Marginals} & \begin{tabular}{c} Histograms \\ Scatterplots \end{tabular} & \begin{tabular}{c} Mutual Information \\ Correlation \end{tabular}\\ \hline
\textbf{Saturated} & Feature Selection & Feature Relevance \\ \hline
\textbf{Model Structure} & Random Variable Graph & Model Score \\\hline
\end{tabular}
}
\end{center}
\label{table:discover-scs}
\end{table}%

\section{CODED Architecture} \label{sec:sys-arch}

Our system supports two use cases: mixed-initiative and fully automatic. The mixed-initiative design interacts with a human-in-the-loop to support three tasks: 1) Specifying a consistent set of statistical constraints 2) Checking the constraints against the data 3) Repairing datapoints that are not compatible with the constraints. In the fully automatic case, the system checks given input constraints against the data and outputs data points that are likely to be dirty.

Figure~\ref{fig:coded} shows the architecture of our  \SC-based error-detection system. \system takes as input a dirty dataset and a set of \SCs (see Section~\ref{sec:sc-discovery}).
The system consists of three key components: \emph{consistency checking}, \emph{hypothesis testing}, and \emph{error drill down}. 
An overview of each component follows. 


\vspace{.25em}

{\noindent \bf Consistency Checking.}  Checking whether the user-specified constraints are consistent helps the user to avoid specification errors \cite{bohannon2007conditional}. The problem of deciding whether a given set of statistical constraints is consistent  has been studied extensively by AI researchers 
~\cite{Niepert2013,pearl2,studeny1990conditional}. We adopt this work, and give improvements for their method in certain conditions (see Section~\ref{cck}). If our method detects that the given \SCs are mutually inconsistent, it helps the user to resolve the conflict. For example, consider the four \SCs in Figure~\ref{fig:running-example}. As will be explained in Section~\ref{cck}, they are actually inconsistent, i.e., $\SC_2$ and $\SC_3$ implies $\textsf{Price} \indep \textsf{Fuel}$, which contradicts $\SC_4$. \system first tells the user that $\{\SC_1, \SC_2, \SC_3, \SC_4\}$ are inconsistent, and then guides the user to resolve the conflict, e.g., removing $\SC_2$. 


\begin{figure}[t]
   \centering
   \includegraphics[width=1\linewidth]{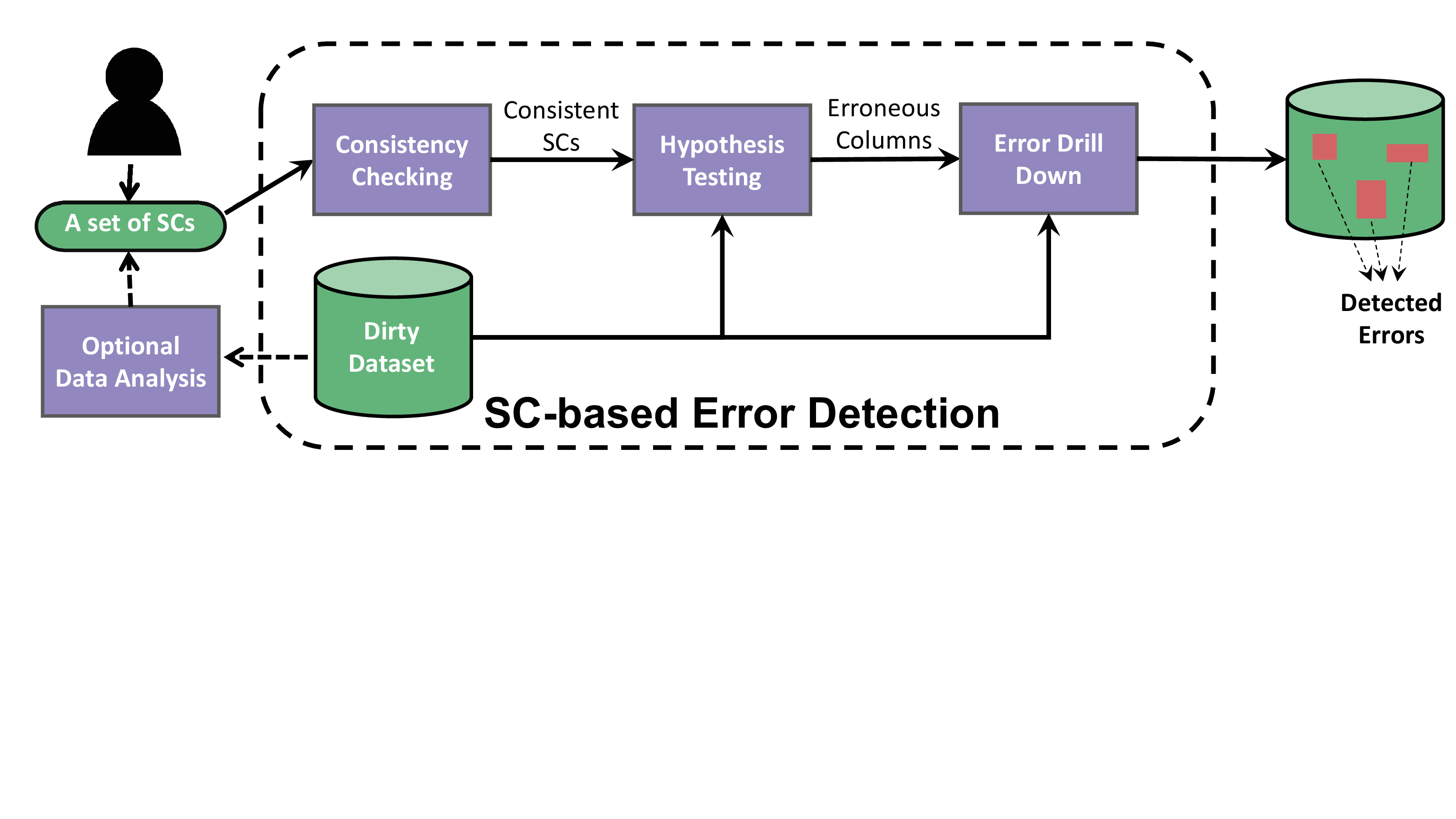} \vspace{-7em}
   \caption{\system Architecture}\label{fig:coded}\vspace{-1.5em}
\end{figure}

\vspace{.25em}

{\noindent \bf Hypothesis Testing.} 
In the hypothesis testing stage, the system checks for each \SC, whether the dataset violates the \SC or not. For example, consider $SC_1 = \textsf{Model} \indep \textsf{Color}$ and the dataset in Figure~\ref{fig:running-example}. The user expects that $\textsf{Model}$ and $\textsf{Color}$ should have an independent relationship. However, a $\chi^2$ hypothesis test on the data indicates that $\textsf{Model}$ and $\textsf{Color}$ are highly correlated. The user then has a choice to revise the constraint, or to maintain the constraint and have the system  mark $\textsf{Model}$ and $\textsf{Color}$ as erroneous columns. Section~\ref{sec:detection} below discusses hypothesis testing process.

 \vspace{.25em}

%
{\noindent \bf Error Drill Down.} If a dataset violates an \SC, the user may want to  drill down into individual records so that she can understand why an \SC is violated and repair the datapoints that caused the violation. Therefore \system provides an error-drill-down component. The user chooses a violated \SC and specifies a threshold $k$. This component returns the top-$k$ records that contributes the most to the \SC's violation. Suppose $\SC_2 = \textsf{Price} \nindep \textsf{Fuel}$ is violated. As will be shown in Section~\ref{sec:drill-down}, if the user specifies $k = 1$, the system will return $r_4$, which is a cheap car with high fuel consumption. The user can examine this record and may find out that unlike the other cars, this is a used car. It is worth noting that given $\SC_4$ the user does not need to write a handcrafted rule like \emph{a cheaper car tends to be more fuel-efficient}.
\section{SC Violation Detection}~\label{sec:detection} We study how to detect an \SC violation in this section. We first reduce to evaluating marginal \SCs in Section~\ref{sec:violation-all-types}, and then discuss how to apply hypothesis testing to detect the violation of marginal \SCs in Section~\ref{subsec:ht}. 




\subsection{Reduction to Evaluating Marginal \SCs}~\label{sec:violation-all-types}
Evaluating a conditional  \SC against the data can be reduced to evaluating marginal \SCs via the following steps. 1) Reduce testing a complex independence $\set{\Xvariable} \indep \set{\Yvariable} | \set{\Zvariable}$ to testing a set of elementary \SCs $\Xvariable \indep \Yvariable | \set{\Zvariable}$, for each $\Xvariable \in \set{\Xvariable}$ and $\Yvariable \in \set{\Yvariable}$. 2) If all conditioning variables are categorical, there is a finite set of possible combinations $\set{\Zvariable}=\set{\zvalue}$; conditional independence is equivalent to marginal independence $\Xvariable \indep \Yvariable | \set{\Zvariable} = \set{\zvalue}$ for each $\set{\zvalue}$ (Definition~\ref{scd}). If $\set{\Zvariable}$ contains continuous variables, we deploy a state-of-the-art \emph{data bucketing} method in statistics to {\em discretize} the data into categorical types with histogram. Discretization groups a number of continuous values into a small number of bins. CODED is open to all the data bucketing methods. Our experiments deployed  the methods 
introduced in~\cite{scott2015multivariate}. The reported results use the bucketing methods with the best performance (f-score). 

{\em Justification.} 
Step 1 can be justified as follows. The decomposition axiom (see Definition~\ref{def:graphoid-axioms}) shows that in any joint probability distribution, a complex SC entails the corresponding set of pairwise elementary SCs, so this direction entails no loss of generality. The converse is the composition principle, that elementary independencies between single variables can be combined into complex independencies between sets of variables. While this is not true for every probability distribution,  
Chickering and Meek \cite{Chickering2002}  show that the composition property holds for a large
class of probability distributions, and provide evidence that real-world datasets satisfy it. Niepert {\em et al.} also give formal conditions under which the elementary and general SCs are equivalent \cite{Niepert2013}. We can therefore expect that testing complex \SCs by decomposing them into elementary \SCs will be adequate for error detection in practice.

To evaluate marginal \SCs $\Xvariable \indep \Yvariable$ for both $\Xvariable$ and $\Yvariable$ either categorical or numerical, we propose default hypothesis testing methods in the remainder of this section. For the mixed case where $\Xvariable$ is numerical and $\Yvariable$ discrete, we can proceed as follows. For simplicity, assume that $\Yvariable$ is binary 0 or 1. Then we have two continuous conditional distributions $P(\Xvariable|\Yvariable=0)$ and $P(\Xvariable|\Yvariable=1)$. The independence between $\Xvariable$ and $\Yvariable$ can then be evaluated by testing whether the two conditional distributions are significantly different~\cite{uther1998tree}. 

\vspace{-.5em}

\subsection{Hypothesis Testing}\label{subsec:ht}
A \emph{hypothesis test} $\test$ is a procedure that takes as input a dataset $\dset$ and outputs either 0 (``the hypothesis is rejected") or 1 (``the hypothesis is not rejected"). A statistical hypothesis is rejected when the probability of the data entailed by the hypothesis is below a user-specified threshold. Most hypothesis tests are based on a {\em test statistic} $\statistic{\dset}$, which returns a real number for a dataset. For independence tests, intuitively, the statistic is an aggregate function that summarizes the degree to which the dataset violates the independence hypothesis $\set{\Xvariable} \indep \set{\Yvariable} | \set{\Zvariable}$. The \emph{$p$-value} specifies the probability of observing a value at least as great as the test statistic for the dataset, assuming the independence constraint: $p(\dset) = P(t \geq \statistic{\dset}|\set{\Xvariable} \indep \set{\Yvariable} | \set{\Zvariable}).$ 

CODED can deploy any hypothesis tests that specified by users, while including default hypothesis testing methods. To make the paper self-contained, we briefly review the basics of these hypothesis testing methods. 

\vspace{.5em}

{\noindent \textbf{The $\chi^2$ test}} 
measures the (in)dependence among categorical data~\cite{pearson1900x}. Given an assignment $\set{\Xvariable} = \set{\xvalue}$, each datapoint satisfies the assignment or not. The \emph{observed count} is the number $N_{\dset}(\set{\Xvariable} = \set{\xvalue})$ of datapoints that satisfy it. 
If the sets $\set{\Xvariable}$ and $\set{\Yvariable}$ were completely independent, the observed counts would equal the product of the marginal counts; these products are called the \emph{expected counts}: $E_{\dset}(\set{\Xvariable} = \set{\xvalue},\set{\Yvariable} = \set{\yvalue}) \equiv N_{\dset}(\set{\Xvariable} = \set{\xvalue}) \times N_{\dset}(\set{\Yvariable} = \set{\yvalue}).$
The $\chi^2$ statistic is calculated as:
\vspace{-.5em}
$$\cstat(\dset) \equiv  \sum_{\set{\xvalue},\set{\yvalue}} \frac{[N_{\dset}(\set{\Xvariable} = \set{\xvalue},\set{\Yvariable} = \set{\yvalue}) - E_{\dset}(\set{\Xvariable} = \set{\xvalue},\set{\Yvariable}= \set{\yvalue})]^2}{E_{\dset}(\set{\Xvariable} = \set{\xvalue},\set{\Yvariable}= \set{\yvalue})}.$$
\vspace{-.5em}

The distribution of the chi-square statistic can be approximated by the $\chi^2$ distribution. That is, we have $p({\dset}) \approx \int_{Q(\dset)}^{+\infty} \chi_{k}^2(t) dt$ where $k = (r_{X}-1) \times (r_{Y}-1) $ and $r_X$ resp. $r_Y$ is the number of possible assignments for $\set{X}$ resp. $\set{Y}$. We use the $\chi^2$ statistic for testing independence among discrete (categorical).
\vspace{.5em}

{\noindent \textbf{Rank Correlations} for numeric variables are based on a straightforward intuition: Each variable defines an ordering over data points, and if the variables are associated, the $\Xvariable$-ordering should carry information about the $\Yvariable$-ordering. 
A number of similarity metrics for rankings have been proposed; In this paper, we choose the the most common non-parametric metric Kendall's $\tau$ as the default method, defined as follows. 
Consider $n$ datapoints with two features $\dset= (\xvalue_1,\yvalue_1), \ldots,$
$(\xvalue_n,\yvalue_n)$. For two datapoints $(\xvalue_i,\yvalue_i)$ and $(\xvalue_j,\yvalue_j)$ with $i \neq j$, if $\xvalue_i > \xvalue_j$ and $\yvalue_i > \yvalue_j$, or $\xvalue_i < \xvalue_j$ and $\yvalue_i < \yvalue_j$, then the two variables agree on the ordering of $i$ and $j$ and the pair $(i,j)$ is \emph{concordant}. The number of concordant pairs is denoted as $n_c(\dset)$. Else if $\xvalue_i > \xvalue_j$ and $\yvalue_i < \yvalue_j$, or $\xvalue_i < \xvalue_j$ and $\yvalue_i > \yvalue_j$, then the two variables disagree on the ordering of $i$ and $j$ and the pair $(i,j)$ is \emph{discordant}. The number of discordant pairs is denoted as $n_d(\dset)$. Pairs neither concordant  nor discordant are called {\em tied}, and their number is denoted as $n_t(\dset)$.
The $\tau$ statistic is then computed as 
\vspace{-.5em}
$$\tau(\dset) \equiv \frac{n_c(\dset) - n_d(\dset)}{\binom{n}{2}}.$$
\vspace{-.5em}


To compute $p$-values, statistical applications use a rescaled denominator, still a function of only $n$, such that the rescaled statistic has an approximately Gaussian distribution. 


\vspace{.5em}

\noindent \textbf{Comparison}. There are many statistics methods \cite{wasserman2013all} that can be applied to verify the independence. We explain why we chose Kendall's $\tau$ as a default method chosen over other popular options (e.g., Pearson's coefficient $\rho$ and Spearman's $\rho_{s}$). 1) The default method in CODED should be compatible with many data characteristicss, so the fewer underlying data assumptions, the better. In statistical terminology, this means we want to use a \emph{non-parametric} hypothesis test, which does not that the dependence can be characterized by a fixed set of parameters known a priori before data observation. 
Pearson's $\rho$ is a \emph{parametric} method that measures the degree to which $\Xvariable$ and $\Yvariable$ are linearly related. The disadvantage of the $\rho$ statistic is that it is reliable only under certain assumptions, including that the relationship between $\Xvariable$ and $\Yvariable$ is approximately linear. The computation complexity of Spearman's $\rho_{s}$ is relatively smaller than that of Kendall's $\tau$. However, comparison studies \cite{knight1966computer,xu2013comparative,howell2009statistical,croux2010influence,fredricks2007relationship} have found that Kendall's $\tau$ is generally more robust in avoiding false positives, which makes it preferred for data error detection.
\section{Error Drill Down}\label{sec:drill-down}

In this section, we study how to drill down into individual records that contribute the most to  the violation of an \SC. We first propose a general framework in Section~\ref{subsec:framework} and then propose two efficient top-k algorithms for categorical data and numerical data, respectively, in Section~\ref{subsec:top-k-algo}.






\subsection{Error-Drill-Down Framework}\label{subsec:framework}
There are two reasons that motivate the drill-down task. Firstly, an \SC violation helps a user to detect which columns have errors, but it does not tell the user which values in the columns may contain the errors. For example, suppose a user specifies an \SC: $\textsf{\small Model} \indep \textsf{\small Color}$ on a car dataset. After a hypothesis test, she finds that $\textsf{\small Model}$ and $\textsf{\small Color}$ violate the independent relationship on the dataset. At this point, she knows that there may be some errors in the $\textsf{\small Model}$ and $\textsf{\small Color}$ columns. We want to help her to locate the errors in the columns so that she can figure out why the violation happens and then fix the errors.  

Secondly, sometimes a dataset has errors, but the errors are not of sufficient magnitude to violate a (in)dependence   relationship among columns. We want to leverage \SCs to detect the errors even in this situation. For example, consider an \SC: $\textsf{\small Fuel} \nindep \textsf{\small Price}$. Suppose that after a hypothesis test, the test result indicates that that $\textsf{\small Fuel}$ and $\textsf{\small Price}$ are dependent (i.e., no violation) but the dependent relationship is weaker than what she expects. We want to help her investigate why the correlation is weaker than expected.

To this end, we develop an interactive error-drill-down framework. A user specifies an \SC and a hypothesis testing method (e.g., $\tau$ test). The framework first applies the hypothesis test to the data, and checks whether \SC is violated. If yes, it will return $k$ records whose column values are most likely to cause the violation. If no, a user can check whether the returned statistic (e.g., the $\tau$ statistic) is larger or smaller than her expectation. Suppose the $\tau$ statistic is smaller than what she expects. Then, she can use the framework to identify $k$ records whose column values are most likely to cause the unexpected result.

Central to this framework is the \emph{top-k error detection problem}. Let $D$ denote a dataset, and $S$ denote a hypothesis testing statistic (e.g., the $\tau$ statistic or the $\chi^2$ statistic). Without loss of generality, we consider only the case where the statistic is larger than a user's expectation. In this situation, the goal is to find $k$ records from $D$ such that if they were removed, the test statistic would decrease the most. Definition~\ref{def:top-k} formally defines the problem.

\sloppy

\begin{Definition}[Top-$k$ Individual Error Detection]\label{def:top-k} Given a dataset $D$, an \SC, a hypothesis testing statistic $S$, and a threshold $k$, we aim to identify a set of $k$ records, denoted by $\Delta D$, such that $S(D-\Delta D)$ is mimized, i.e., 
\[ \argmin_{\Delta D \subseteq D} \,\, {S(D-\Delta D)}  \quad \textsf{s.t.} \quad |\Delta D| = k \]
\end{Definition}\vspace{-.5em}


A naive solution is to enumerate all $\binom{|D|}{k}$ possibilities, and then return the best result. This is prohibitively expensive. Even with a modest data size of 10,000 records, assuming a reasonable k = 10, this approach would require enumerating $\binom{10000}{10} = 2.7\times 10^{33}$ possibilities. Therefore, we adopt greedy algorithms to reduce the cost. We propose two greedy algorithms, \emph{$K$ strategy} and \emph{$K^c$ strategy}. Intuitively, the $K$ strategy seeks to directly identify the best $k$ records; the $K^c$ strategy seeks to remove the worst $n-k$ records and then return the remaining $k$ records as a result. 

\vspace{.25em}

\noindent \textbf{\emph{$K$} Strategy.} The algorithm first selects the \emph{best} record $d^*$ from $D$ such that if it was removed, the statistic can decrease the most. The algorithm removes $d^*$ from $D$, and then repeats the above process to select the best record from  $D - \{d^*\}$. After $k$ iterations, the top-$k$ records are identified.

\vspace{.25em}

\noindent \textbf{\emph{$K^c$} Strategy.} The algorithm first selects the \emph{worst} record $d'$ from $D$ such that if it was removed, the statistic can \emph{increase} the most. The algorithm removes $d'$ from $D$, and then repeats the above process to select the worse record from  $D - \{d'\}$. After $n-k$ iterations, where $n = |D|$, the remaining $k$ records are returned.

\vspace{.25em}

{\noindent \bf Remark.}
The \emph{$K$} strategy is more efficient than the \emph{$K^c$} strategy because the former only needs to select $k$ records but the latter needs to check $n-k$ records. In terms of effectiveness, the \emph{$K$} strategy often leads to a better objective value (i.e., smaller $S(D-\Delta D)$ because it directly optimizes for that value. The \emph{$K^c$} strategy is particularly useful in identifying a set of $k$ records that are highly correlated with each other, thus it is more suitable to detect errors for the violation of an independence SC.

\subsection{Top-k Error Detection Algorithms}\label{subsec:top-k-algo}

We describe top-$k$ error detection methods for the two statistics we examine in this paper, $\chi^2$ and $\tau$. We discuss how to implement the $K$ strategy for them efficiently. The same implementation can be extended to the $K^c$ strategy trivially.




\subsubsection{Categorical Data}

\begin{figure}[t]
   	\centering
   	\includegraphics[width=0.55\linewidth]{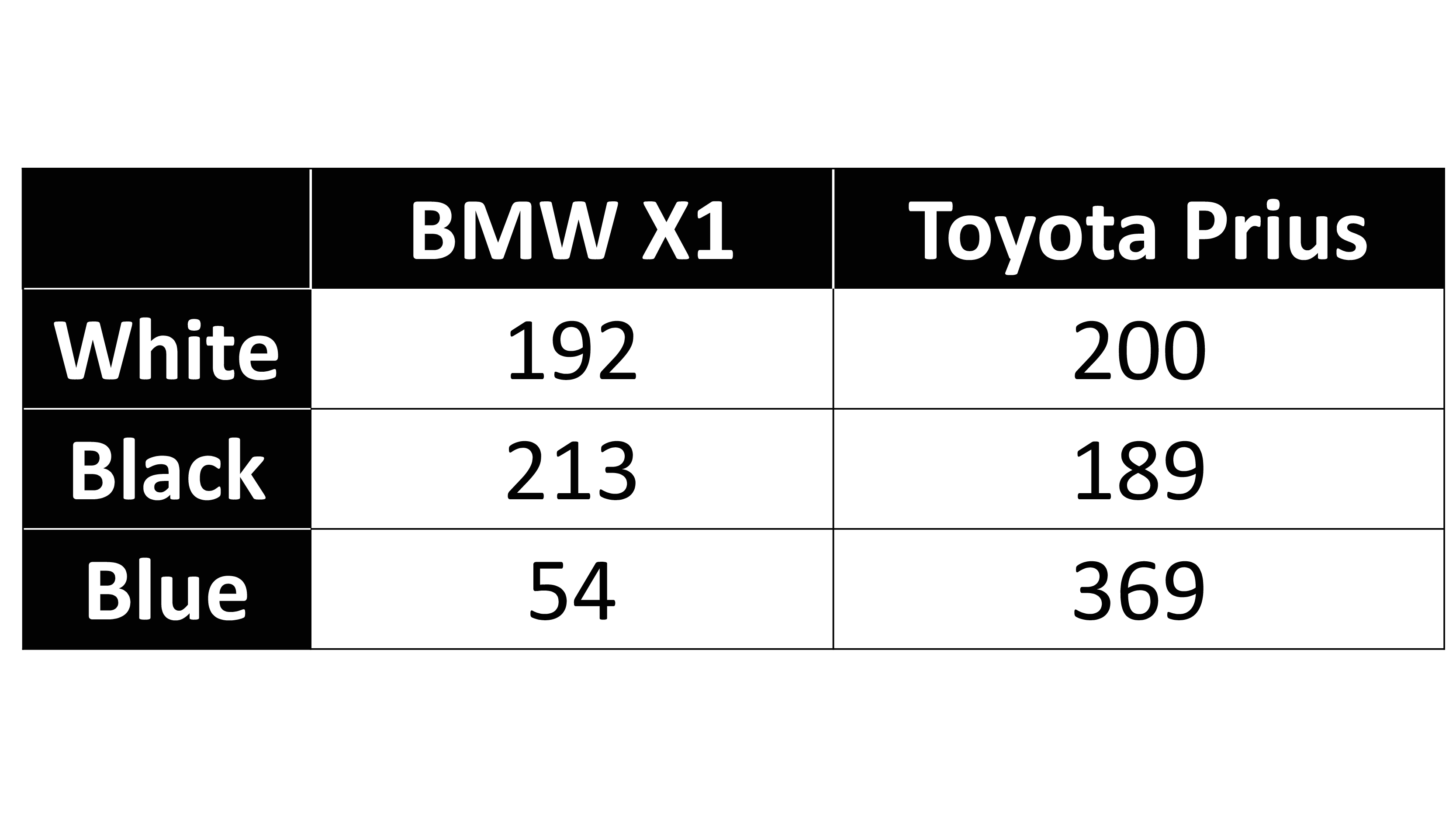}
   	\vspace{-2em}
   	\caption{Group counts of \textsf{Model} and \textsf{Color}}\label{fig:car_cat}
   	\vspace{-1em}
\end{figure}

We use $\chi^2$ test as the default method to detect errors for categorical data. We propose an optimization technique to reduce the time of selecting the ``best'' record at each iteration. The key observation is that if two records have the same values on the tested columns, there is no difference to choose either one of them. Therefore, we can group the records based on the tested columns, and only need to spend time in deciding which group (rather than which record) should be selected at each iteration. 

For example, consider an \SC: $\textsf{\small Model} \indep \textsf{\small Color}$ and a dataset similar to Figure~\ref{fig:running-example} but with more records.
We first group the records based on the tested columns (i.e., $\textsf{\small Model}$ and $\textsf{\small Color}$). Suppose there are 2 car models and 3 colors. Then, there will be 6 groups in total. Figure~\ref{fig:car_cat} illustrates the 6 groups, where the number in each cell represents the total number of the records that belong to that group. Each cell has a $\chi^2$ value  $z$. At each iteration, the $K$ strategy only needs to determine which group should be selected and then randomly picks one record from that group based on the statistic $z$. This optimization technique significantly reduces the computational cost since the number of groups could be orders of magnitude smaller than the total number of records.

\subsubsection{Numerical Data}
We discuss how we employ the framework for numerical data with the $\tau$ test. We formally define the problem as follows.

\begin{Definition}[$\tau$-test-based Error Detection] Given a dataset $D$, an \SC, and $k$, the $\tau$-test-based error detection problem tries to find a subset $\Delta D$ of records from $D$ that contribute the most to the violation of the \SC, i.e.,
\[ \argmin_{\Delta D \subseteq D} \,\, \frac{n_c(\dset - \Delta D) - n_d(\dset- \Delta D)}{\binom{|\dset| - k}{2}}  \quad \textsf{s.t.} \quad |\Delta D| = k \]
\end{Definition}

For simplicity, we denote the count of concordant pairs, discordant pairs, and tied pairs in $\dset - \Delta D$ as $n_c$, $n_d$, and $n_t$, respectively. We omit the denominator of the objective function since it is a constant function of $n$ only. 
%
The $K$ strategy works as follows. At each iteration, it calculates the \emph{benefit} of each record: Given a record $r$, find all the pairs that contain the record, and then calculate the sum of the weights of these pairs, denoted by $\textsf{benefit}(r)$. 
Then, we select the record with the biggest benefit. Once a record is selected, we need to update the benefit of all the remaining records. We use a priority queue to maintain the top-$k$ records. Algorithm~\ref{alg:tau} shows the pseudo-code. Example~\ref{exa:tau-algo} illustrates the algorithm. 


\begin{algorithm}[t]
    \caption{\small $\tau$-test-based error detection algorithm}\label{alg:tau}
    \KwIn{An \SC, Dataset $D$, $k$}
    \KwOut{$k$ records}
    $Q \leftarrow \emptyset$; // {\scriptsize Priority Queue} \\
    $R \leftarrow \emptyset$; // {\scriptsize Returned List} \\
    $W \leftarrow \emptyset$; // {\scriptsize A hash table that maps a pair to its weight} \\
    $\textsf{benefit}(r) = 0$ for $r \in D$ \;
    
     // Initialization \\
    \For{$r_i \in D$}
    {
    \For{$r_{i+1} \in D$}
    {
     \If{$r_i, r_j$ \emph{are concordant}}
     {
       $\textsf{benefit}(r_i)$ += 2; ~~ $\textsf{benefit}(r_{i+1})$ += 2; ~~ $W(r_i,r_{i+1}) = 2$;

     }
     \ElseIf{$r_i, r_j$ \emph{are tied}}
     {
       $\textsf{benefit}(r_i)$ += 1;~~ $\textsf{benefit}(r_{i+1})$ += 1; ~~$W(r_i,r_{i+1}) = 1$;

     }
     $Q.push(\langle r_i, \textsf{benefit}(r_i) \rangle)$; 
    }
    }
    
    // Iteration \\
    \For{$i = 1$ to $k$}{
        Add $Q.top()$ to $R$\;
        Update $Q;$ // {\scriptsize Update the benefit of each record $r \in Q$ that has a non-zero weight with $r_i$ ($W(r,r_{i})=2~\textrm{or}~1$)} \\
    
    }
   {\bf return} R\;
\end{algorithm}
\setlength{\textfloatsep}{.5em}




\begin{example}\label{exa:tau-algo}
Following the running example in Figure~\ref{fig:running-example}, we apply the $tau$ testing method to detect data errors w.r.t. $\SC_1 =\textsf{Price} \nindep \textsf{Fuel}$, and let $k=1$.

\noindent\underline{Step 1: Initialization.} We calculate the benefit of each record on the entire dataset, and obtain the priority queue $Q=\{r_4:0, r_5:9, r_7:9, r_1:10, r_2:10, r_8:10, r_3:12\}$. Intuitively, the benefit of a record (e.g., $\textsf{benefit}(r_5) = 9$) means that if $r_5$ was removed, the objective value would be decreased by 9.    


\noindent\underline{Step 2: Removing and Updating.} Return the top element of $Q$, in this case $r_4$. Then remove $r_4$, and update the priority queue. Since $r_4$ shares a zero weight with every other record, the benefit of each record keeps unchanged. The updated queue will be $Q=\{r_5:9, r_7:9, r_1:10, r_2:10, r_8:10, r_3:12\}$. 

\noindent\underline{Step 3: Repeating Step 2.} We repeat Step 2 until $k$ records are returned. Here, since $k = 1$, only $r_4$ will be returned. 

\end{example}

\noindent \textbf{Efficiency Analysis.} The main computational bottleneck of the $K$ strategy is the initialization phase. Consider a dataset with two columns: $D = \langle x_1, y_1 \rangle, \cdots, \langle x_n, y_n\rangle$. We need to initialize the benefit of each record. For a single record, this requires comparing with the other $n-1$ records, leading to a time complexity of $\mathcal{O}(n)$. Therefore the naive implementation of the initialization phase needs $\mathcal{O}(n^2)$ time. This does not scale to large datasets (e.g. 1M records). 

The time complexity of the initialization phase can be reduced to $\mathcal{O}(n \log{n})$ with a segment tree. A segment tree is a tree data structure, where each node stores information about a segment. It allows for inserting a segment and querying a segment with both $\mathcal{O}(\log{n})$ time. To apply this idea, we first sort $D$ by column $X$ and then scan the records in $D$ based on this new order. For each record $\langle x_i, y_i \rangle$, we can get the number of concordant pairs of the record by querying the segment of $(-\infty, y_i)$, and the number of discordant pairs by querying the segment of $(y_i, +\infty)$. Once the two numbers are obtained, we insert a segment $[y_i, y_i]$ (representing a single point) to the tree. We need $\mathcal{O}(\log{n})$ time to process each record, thus the total time complexity is $\mathcal{O}(n \log{n})$. With this optimization, error drill-down can handle a dataset with millions of records. Further details are in Appendix~\ref{app:efficient-k}.

\section{Consistency Checking} \label{cck}

In this section, we first define the consistency checking problem, then discuss different approaches for different types of constraints.
Intuitively, if a set of \SCs $\Sigma$ is consistent, it means that $\Sigma$ has no conflict. That is, there exists a joint distribution that satisfies $\Sigma$. 

\begin{Definition}[Consistency]\label{def:consistency}
Given a set $\Sigma$ of \SCs, let $\set{V}$ denote the set of random variables that appear in $\Sigma$. The consistency problem is to determine whether there exists a joint distribution $P(\set{V})$ that satisfies every \SC in $\Sigma$. 
\end{Definition}

\begin{figure}[t]
   \centering 
   \includegraphics[width=0.9\linewidth]{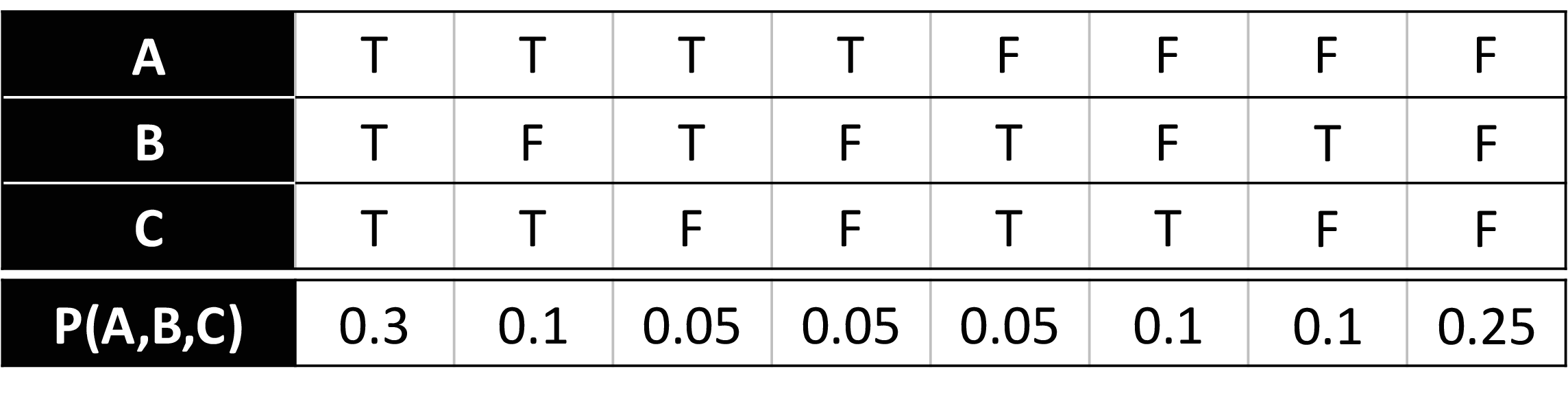} 
   \vspace{-1em}
   \caption{An example of a joint distribution that satisfies $\Sigma = \{A \indep B,\; A \nindep B| C\}$}\label{fig:joint}
\end{figure}

For example, consider $\Sigma = \{A \indep B,\; A \nindep B| C\}$. Since there are three random variables in $\Sigma$, we have $\set{V} = \{A, B, C\}$. As shown in Figure~\ref{fig:joint}, we can construct a joint distribution $P(A, B, C)$ that satisfies $\Sigma$, thus $\Sigma$ is consistent. If we are given $\Sigma' = \{A \indep B,\; A \nindep B\}$, it is impossible to construct a joint distribution that satisfies $\Sigma'$, thus $\Sigma'$ is inconsistent. 


For complete models and elementary saturated constraints (Section~\ref{sec:sc-discovery}), consistency checking is easy. (1)  
\SCs represented in a {\em complete model structure} is {\em guaranteed to be consistent} because parametrizing the elicited model defines a probability distribution 
that satisfies the constraints. (2) For \SCs statements that are saturated and elementary,
Heckerman et al. show that a set of \SCs is consistent if and only if it satisfies the symmetry axiom described in Section~\ref{def:graphoid-axioms} \cite{Heckerman2000}, which is easy to check. We therefore focus our discussion on general \SCs.

\subsection{Consistency Checking By Inference} \label{subsec:consistency}

Determining the consistency of general \SC statements is a very difficult problem; it is not even known whether it is decidable \cite{Niepert2013}. Nonetheless,  AI researchers have developed a number of heuristic approaches~\cite{Niepert2013,pearl2,studeny1990conditional}.
A principled approach is to use an \emph{inference framework} for deriving a contradiction~\cite{pearl2}, as illustrated in Algorithm~\ref{alg:check}.  A similar inference framework  is used to check the consistency of a set of functional dependencies \cite{bohannon2007conditional}. 
For example, consider $\Sigma = \{A \indep B,\; A \nindep B| C\}$. We have $\mathcal{I} = \{A \indep B\}$  and $\mathcal{D} = \{A \nindep B| C\}$. We can see that $\mathcal{D}$ does not imply  $\lnot \lambda = A \indep B| C$, thus $\Sigma$ is consistent. 

Once an inconsistency is derived by implication, {\em the derivation can help the user to resolve the inconsistency.}  Specifically, suppose that the consistency-checking algorithm finds a conflict: $\mathcal{I}$ implies $A \indep B$ but $A \nindep B$ is in $\Sigma$. We ask the user to check which one is correct. If the former is correct, 
we remove $A \nindep B$ from $\Sigma$. Otherwise, we show the derivation of $A \indep B$  from $\mathcal{I}$ and then ask the user to remove the incorrect \SCs involved in the derivation. The implication approach requires an inference system for Independence SCs, which we describe next.

\begin{algorithm}[t]
	\caption{\small Consistency-checking By Implication}\label{alg:check}
	\KwIn{A set of \SCs: $\Sigma$}
	\KwOut{Consistent or not (YES or NO)}
	\vspace{0.5em}
	Divide $\Sigma$ into independencies $\mathcal{I}$ and dependencies $\mathcal{D}$\;
	Iteratively apply inference to $\mathcal{I}$ until a fixed point $\mathcal{I}^*$\;
	\For{$\lambda \in \mathcal{D}$}
	{
	    \If{$\lnot \lambda \in \mathcal{I}^*$}
	    {
	        {\bf return} NO\; 
	    }
	}
	{\bf return} YES\;
\end{algorithm}
\setlength{\textfloatsep}{.5em}

\vspace{.5em}

{\noindent \bf The Graphoid Axioms for Independence SCs.} A commonly used set of axioms is Geiger and Pearl's system $\graphoid$~\cite{pearl2}.



\begin{Definition}\label{def:graphoid-axioms}
Let $\set{X}$, $\set{Y}$, $\set{Z}$, and $\set{W}$ denote four disjoint sets of random variables. The \emph{graphoid axioms} $\graphoid$ are the following:

\vspace{.5em}

{\noindent Symmetry}: \small{$\set{X} \indep \set{Y} \mid \set{Z}  \quad \Longleftrightarrow \quad  \set{Y} \indep \set{X} \mid \set{Z}$}

{\noindent Decomposition}: 	\small{$\set{X} \indep \set{YW} \mid \set{Z} ~~\Longrightarrow~~ \set{X} \indep \set{Y} \mid \set{Z} ~~ \& ~~  \set{X} \indep \set{W} \mid \set{Z}$}

{\noindent Weak Union}: \small{$\set{X} \indep \set{YW} \mid \set{Z}  ~~ \Longrightarrow~~  \quad \set{X} \indep \set{Y} \mid \set{ZW}$}

{\noindent Contraction}: \small{$\big (\set{X} \indep \set{Y} \mid \set{Z} \big ) ~~ \& ~~ \big (\set{X} \indep \set{W} \mid \set{YZ} \big ) ~~ \Longrightarrow ~~  \set{X} \indep \set{YW} \mid \set{Z}$}

\end{Definition}
\noindent Note that $\set{Z}$ can be the empty set.
For example, for the symmetry axiom, we have $\set{X} \indep \set{Y}   \Longrightarrow \set{Y} \indep \set{X}.$ 

\vspace{.5em}

\noindent \emph{Example.} Suppose we are given three constraints:
\[\Sigma = \{\SC_1:  A \indep B, ~~~~\SC_2:  A \indep C | B, ~~~~\SC_3: A \nindep C\}.\] 
Applying the Contraction rule to $\SC_1$ and $\SC_2$:
\[(A \indep B) ~~\&~~  (A \indep C | B) \quad  \Longrightarrow \quad A \indep BC.\]
Applying the Decomposition rule to $A \indep BC$:
\[(A \indep BC) ~~\Longrightarrow~~  A \indep B ~~\&~~  A \indep C.\]
Then, we can see that $A \indep C$ is inconsistent with the given $\SC_3$. Therefore, $\Sigma$ is inconsistent.









\vspace{.5em}

{\noindent \bf Soundness and Completeness.}
Traditional functional dependencies (FDs) have a sound and complete inference system that is known as Armstrong's Axioms~\cite{abiteboul1995foundations}. Wenfei et al. extend Armstrong's Axioms to CFDs and prove that the new inference system is also sound and complete~\cite{bohannon2007conditional}, but deciding the consistency of a set of CFDs is NP-hard.
Unfortunately, inference for \SCs is a hard problem: a major negative result due to Studeny says that no finite set of axioms is both sound and complete for general \SCs \cite{studeny1990conditional}.  The system $\graphoid$ is sound~\cite{pearl2} (for strictly positive joint probability distributions), but by Studeny's result, not complete. The situation for general \SCs is thus comparable to that for DCs, where we currently have a set of sound (but not complete) inference rules~\cite{chu2013holistic}. Niepert et al. introduce another set of axioms $\niepert$ for \SCs, which is complete but not sound \cite{Niepert2013}. They suggest applying both $\graphoid$ and $\niepert$ to a given set of constraints, which allows the implication algorithm to both falsify some cases of inconsistent \SCs (via $\graphoid$) and validate some cases of consistent SCs (via $\niepert$). We leave evaluating system $\niepert$ for future work. 





\vspace{.5em}

{\noindent \bf Time Complexity.} 
We next show that the inference method based on system $\graphoid$ is efficient if the number of variables that appear in any of the input \SCs is relatively small. 
%
To analyze the time complexity of the implication algorithm, we answer two  questions:\\
\indent \emph{Q1. What is the time complexity of generating a new \SC?} \\
\indent \emph{Q2. How many new \SCs can be generated in total?}


For $Q1$, consider the two inference rules in Definition~\ref{def:graphoid-axioms}. For the rules with a single left-hand \SC, the algorithm can check each derived \SC for whether it can be used to generate a new \SC. The time complexity of this process is $\mathcal{O}(|\mathcal{I}^*|)$. For the rules with two left-hand \SCs, 
a simple approach is to enumerate every pair of \SCs in $\mathcal{I}^*$ and then check whether the rule can be used or not, requiring $\mathcal{O}(|\mathcal{I}^*|^2)$ time. 
In total, our algorithm needs $\mathcal{O}(|\mathcal{I}^*|^2)$ time to generate a new \SC at each iteration.      


For $Q2$, we seek to compute an upper bound for the number of newly generated \SCs. The key observation is that {\em for any inference rule in $\mathcal{G}$, every variable that occurs on the right-hand SC also occurs on the left-hand side.} For example, in the contraction rule, four sets of variables occur on the left-hand side, and the same four sets on the right-hand side. We say $\SC_1$ \emph{covers} $\SC_2$ if the variable set of $\SC_1$ is a superset of the variable set of $\SC_2$. For a fixed set of $m$ variables, there are $O(3^m)$ possible \SCs with these $m$ variable (depending on whether each variable is assigned to the first set, the second set, or the conditioning set). Given the key observation, we can therefore bound the number of $\SCs$ generated by a single $\SC$ as 
$O(\sum_{m=2}^{\ell} \binom{\ell}{m} \times 3^m.)$ (Each \SC must involve at least two variables.) Overall the number of number of $\SCs$ generated can therefore be bounded by  $O(|\mathcal{I}| \sum_{m=2}^{\ell} \binom{m}{\ell} \times 3^m)$, where $\ell$ is the largest number of variables that occurs in any of the input $\SCs$. This shows that {\em the computational cost of the implication algorithm is pseudo-polynomial in the number of input $\SCs$ $\mathcal{I}$ and the parameter $\ell$, meaning polynomial if $\ell$ is small enough to be treated as a constant~\cite{Garey1979}.} 

\section{Experiments}\label{sec:exp}

\begin{table*}[h]
 \centering \small
 \caption{Constraints used by CODED and DC 
 } \label{tab:con-summaries} \vspace{-1em}
   \begin{tabular}{|l|c|l|}
   \hline
       \textbf{Attributes} & \textbf{CODED} & \textbf{Denial  Constraints} \\ \hline
   \hline
    N\_oxide(N), Distance(D) & $N \nindep D$ & For any two records $r_1$ and $r_2$, if $r_1[N] > r_2[N]$, then $r_1[D] < r_2[D]$\\ \hline
    Rooms(R), Black Index(B) & $R \indep B$ &  $\times$ \\ \hline
  Tax rate, Black Index, Crime(C) & $T \nindep B~ |~C$ & For any $r_1$ and $r_2$ with $r_1[C]=r_2[C]$, if $r_1[T] > r_2[T]$, then $r_1[B] < r_2[B]$\\ \hline
  N\_oxide, Black Index, Tax rate (T) & $N \indep B ~|~T$ &$\times$  \\ \hline

  Buying Price(BP), Class(Cl) & $BP \nindep Cl$ &For any two records $r_1$ and $r_2$, if $r_1[BP] > r_2[BP]$, then $r_1[Cl] > r_2[Cl]$ \\\hline
  Safety(SA), Doors(DR) & $SA \indep DR$ & $\times$  \\\hline
  Temperatures (T) of Sensor~8 and Sensor~9 & $T_8 \nindep T_9$ &  For any two $r_1$ and $r_2$, if $r_1[T_8] > r_2[T_8]$, then $r_1[T_9] > r_2[T_9]$ \\\hline
  Games(G), Goal Plus-Minus(GPM), Year(Y)& $G \indep GPM ~| ~Y$ & $\times$  \\\hline
   \end{tabular}%
 \label{tbl:constraints}\vspace{-.5em}
\end{table*}%

We evaluate the effectiveness and efficiency of our methods on real-life datasets with both synthetic errors and real errors. Specifically, we examine (1) the computational efficiency of the consistency-checking methods, (2) the effectiveness of our method compared to the state-of-the-art approaches on detecting synthetic and real-life errors, and (3) the scalability of our error-detection method.


\subsection{Experiment Setup}
\noindent \textbf{Datasets.} We evaluated our approaches on five real datasets.

{(1) \tt BOSTON\footnote{\scriptsize https://www.cs.toronto.edu/~delve/data/boston/bostonDetail.html}.} The Boston dataset was taken from the Boston Standard Metropolitan Statistical Area (SMSA) in 1970. This dataset was first used in~\cite{Harvard} to study the relationship between clean air quality and household's willing to pay. There are 506 instances, and each instance has 14 attributes. We used $6$ attributes: Distance to CBD area-Distance (D), Nitric Oxides Concentration-N\_oxide (N), Crime Rate-Crime (C), Black index of population(B), Rooms(R) and Tax Rate(T). 

{(2) \tt CAR\footnote{\scriptsize  https://archive.ics.uci.edu/ml/datasets/Car+Evaluation.}.} The Car Evaluation dataset is from UCI Machine Learning repository. This dataset contains seven attributes. We used 4 attributes: Buying price(BP), Car Class (CL), Doors(DR), and Safety level(SA).

{(3) \tt HOSP\footnote{\scriptsize http://www.hospitalcompare.hhs.gov}}. The HOSP dataset contains 100K records with 19 attributes. It was used in previous data-cleaning studies\cite{chu2013holistic,DBLP:journals/pvldb/RekatsinasCIR17}. We got the clean and dirty versions of the dataset from~\cite{chu2013holistic}.  

{(4) \tt SENSOR\footnote{\scriptsize http://db.csail.mit.edu/labdata/labdata.html}.} The Sensor dataset collected the sensor reports from the Berkeley/Intel Lab. The dataset has more than 2 million records, containing the humidity and temperature reports from 54 different sensors. 

{(5) \tt HOCKEY\footnote{\scriptsize https://github.com/liuyejia/Model\_Trees\_Full\_Dataset documents clean and dirty versions.}.} The Hockey dataset collected the records of each NHL game from 1998-2010. It has more than ten attributes which variously describe player attributes and player performance statistics for a season. 


\vspace{.25em}

\noindent \textbf{Simulated Errors.} The purpose of the use of simulated errors is to help us gain a deep understanding of the effectiveness of our approach in various situations. We simulated two types of errors: \emph{sorting error} and \emph{imputation error}. As shown in Section~\ref{sec:intro}, both types of errors have appeared in reality. 

For the sorting error, we selected $\alpha\%$ of column $A$ (randomly or based on column $B$) and sort its values in an ascending order; for the imputation error, we selected $\alpha\%$ of column $A$ (randomly or based on column $B$) and replaced them with the mean value of column A.  $\alpha\%$ is called \emph{error rate}. Note that the sorting error (the imputation error) may either make two columns $A$ and $B$ more independent or less independent based on whether the values are selected randomly or based on column $B$. We used random selection for dependence \SCs, column $B$ for independence \SCs. We also explored the combined impact of the two error types. Our \emph{combination error} consists of 80\% sorting error and 20\% imputation error.

\noindent \textbf{Real-life Errors.} The Sensor and Hockey datasets contain real-life errors. 
To compress the Sensor dataset, we replaced sensor readings by their hourly average. 
The aggregated dataset has many outliers. Since some outliers are very easy to detect, we considered the following scenario to make the experiment more challenging. Consider two data scientists: Alice and Bob. Alice first removed easy-to-detect outliers (e.g., temperature > 100 \textdegree C)  and replaced them with the mean temperature. When Bob got the dataset from Alice, he did not know that the raw dataset was changed by Alice. Imagine Bob wanted to use the dataset differently (e.g., count the outliers in the raw dataset). In this situation,  Bob  wanted to detect not only the remaining outliers but also the outliers removed by Alice. The ground-truth of the dataset was obtained using the method in~\cite{jeffery2006declarative}. 


Hockey is a public dataset used for hockey data analytics. We originally thought that the dataset had no error. However, our hockey knowledge specified an \SC, (last row in Table~\ref{tbl:constraints}), which we noted was violated in the table. 
It turned out that the dataset provider  filled all the missing values with 0 or mean values. 
The correct ground-truth of the dataset was obtained from other reliable hockey websites.

\vspace{.25em}

\noindent \textbf{Error-Detection Approaches.} We compared CODED with three state-of-the-art error detection approaches.


{\tt Denial Constraints (DC)}~\cite{chu2013discovering} is an integrity constraint based error detection approach. The original DC approach does not support top-$k$ error detection. We extended it as follows: for each record $r$, count the number of other records that are inconsistent with $r$ given the DC. Then  return the top-$k$ records that involve the most number of violations. Table~\ref{tbl:constraints} summarizes the constraints used by DC. `$\times$` means that we cannot find a DC to represent the same independence relationship as CODED (cf. Section~\ref{sec:scvsic}).

{\tt DBoost}~\cite{mariet2016outlier} is the state-of-the-art outlier detection approach. This is also used by \cite{abedjan2016detecting} to compare different types of error-detection approaches. 
We used an implementation available online\footnote{\scriptsize https://github.com/cpitclaudel/dBoost}. We applied DBoost with three models: GMM, Gaussian and Histogram. For categorical data, we employed the bin width that achieves the best f-score results. For numeric data, we employed Gaussian and GMM with the mixture parameter n$\_$subpops threshold set at $3, 0.001$, and the statistical epsilon to be~$0$.

{\tt Approximate Functional Dependency(AFD)}~\cite{mandros2017discovering} is an error detection approach based on approximate constraints. 
To make AFD support top-$k$ error detection, we extended it by 
returning the top-$k$ records that lead to the most number of violations. We considered two AFDs on the HOSP dataset \textsf{\small Zipcode -> City} and \textsf{\small Zipcode -> State}, and compared with CODED w.r.t. \textsf{\small Zipcode $\nindep$ City} and \textsf{\small Zipcode $\nindep$ State}.

{\tt CODED} is our \SC-based error detection approach. Table~\ref{tbl:constraints} summarized the constraints used by CODED. For the CODED hypothesis testing, we used the $\chi^2$ test for categorical data, and the $\tau$ test for numerical data. We implemented the error-drill-down framework, and adopted the $K$ strategy for dependence \SCs and the $K^c$ strategy for independence \SCs. 


\vspace{.25em}

\noindent \textbf{Quality Measurement.} We considered two user scenarios.  (i) The user wants to manually examine a small number of records (e.g., $k = 50$) in order to reason about data errors. For this scenario, since $k$ is fixed, we need to maximize \emph{Precision@K}, which is defined as the ratio of the number of correctly detected records to the number k. (ii) The user wants to detect all errors in order to repair them. For this scenario, as $k$ increases, recall increases while precision potentially decreases. We report \emph{Precision@K}, \emph{Recall@K}, and \emph{F-score@K} by varying. \emph{Precision@K} is the same as above. \emph{Recall@K} is the ratio of the number of correctly detected records among the returned $K$ records to the number of total erroneous records, and \emph{F-score@K} is their harmonic mean.


\vspace{-0.75em}
\subsection{Experimental Results}
\noindent \textbf{\underline{{Exp-1: Efficiency of Consistency Checking.}}}
We examined the efficiency of the consistency-checking algorithm. Note that consistency checking is {\em not} related to data. Therefore, randomly generating \SCs is appropriate to test the scalability of our 
consistency-checking algorithm. We randomly generated $|\mathcal{I}|$ independence \SCs and $|\mathcal{D}|$ dependence \SCs containing up to 3 out of $|\mathcal{D}|/2$ variables, where \SCs are elementary CIs (See Table \ref{table:ci-types}).  We measured the runtime of consistency checking by varying $|\mathcal{I}|$ and $|\mathcal{D}|$, respectively. The results are shown in Figure~\ref{fig:scs_check}. Each figure has three lines, which represent different ratios of $|\mathcal{I}|$ to $|\mathcal{D}|$.

We make two interesting observations. First, {\em our consistency-checking algorithm scales well} by varying either $|\mathcal{I}|$ or $|\mathcal{D}|$. Even with 500 independence (dependence) \SCs, the algorithm can terminate within 0.5 ms. Second, {\em the computation time is more related to $|\mathcal{I}|$ than $|\mathcal{D}|$}. This is because the consistency-checking algorithm proposed comprises two parts: implication and checking, where the implication algorithm needs $\mathcal{O}(|\mathcal{I}^2|)$ time, which typically dominates the whole process. 

\begin{figure}[t]\vspace{-1em}
   \centering 
   \includegraphics[width=1\linewidth]{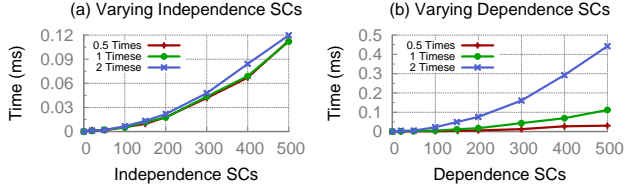} 
   \vspace{-2em}
   \caption{Time for consistency checking of \SCs 
   }\label{fig:scs_check}
\end{figure}

\vspace{.25em}

\begin{figure*}[t]
   \centering 
   \includegraphics[width=0.85\linewidth, height=2.7cm]{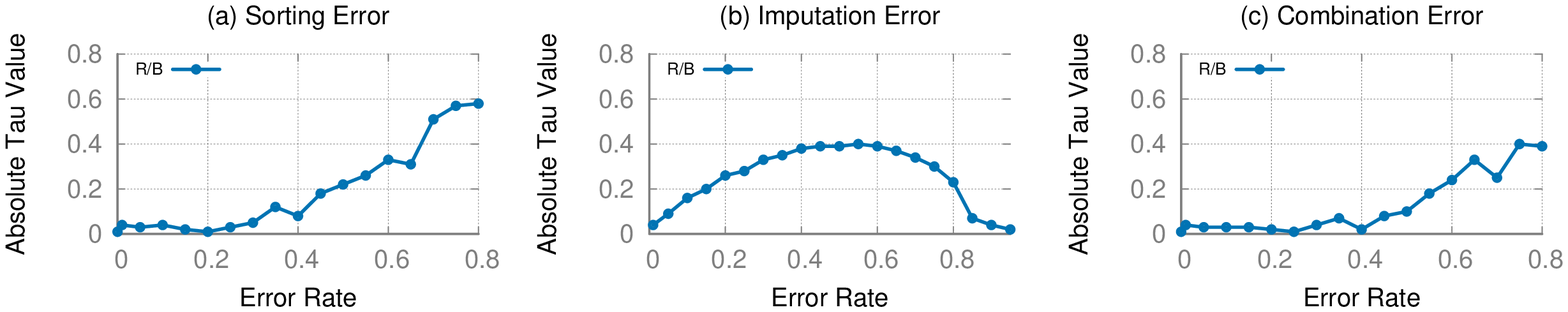} 
   \vspace{-1em}
   \caption{Different error impact on independence \SCs (Boston dataset) 
   }\label{fig:imp_indep}
 \vspace{.5em}
   \centering
   \includegraphics[width=0.85\linewidth, height=2.7cm]{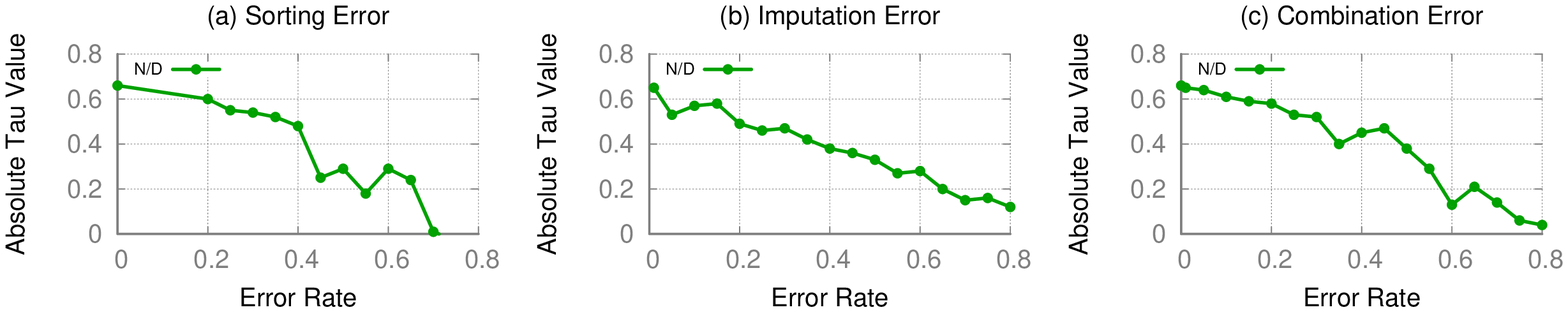} 
   \vspace{-1em}
   \caption{Different error impacts on dependence \SCs (Boston dataset) 
   }\label{fig:imp_dep}\vspace{-.em}
\end{figure*}

\begin{figure*}[t]
   \centering 
   \includegraphics[width=0.85\linewidth, height=2.7cm]{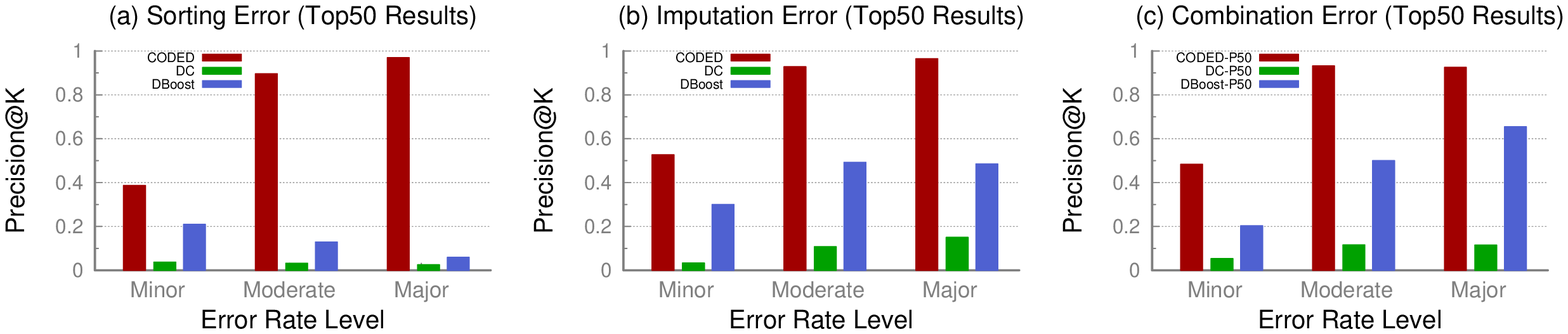} \vspace{-1em}
   \caption{Effectiveness of error detection methods for dependence \SCs (Boston dataset, $K=50$)}\vspace{1em}\label{fig:dep_bar_chart}
   \centering 
   \includegraphics[width=0.85\linewidth, height=2.7cm]{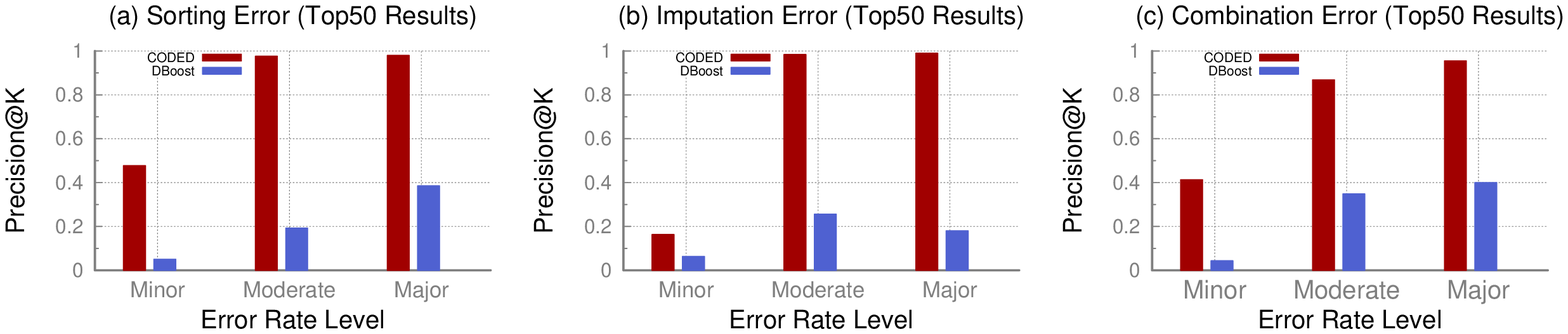}\vspace{-1em}
   \caption{Effectiveness of error detection methods for independence \SCs (Boston dataset, $K=50$)}\label{fig:indep_bar_chart}\vspace{-1em}
\end{figure*}

\begin{figure}[t]
  \centering 
  \includegraphics[width=1\linewidth]{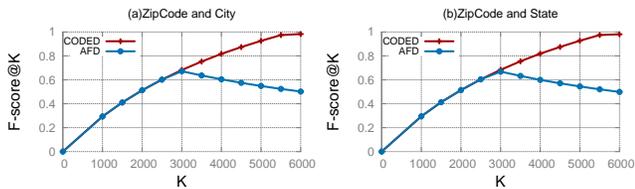} 
  \vspace{-2em}
  \caption{Comparison of CODED and AFD (HOSP)
  }\label{fig:cmp-afd}
\end{figure}

\noindent \textbf{\underline{Exp-2: Impact of Errors.}}
To illustrate the impact of data errors, we varied the error rates of sorting error, imputation error, and combination error, 
and tested how the absolute tau value (i.e., $|\tau|$) changed accordingly. Recall that a larger $|\tau|$ implies a more dependent relationship. We selected two \SCs,  $R \indep B$ (R/B in Figure~\ref{fig:imp_indep} ) and $N\nindep D$ (N/D in Figure~\ref{fig:imp_dep}).

Figure~\ref{fig:imp_indep} shows the results w.r.t. an independence \SC. All errors tend to raise $|\tau|$, indicating a violation of the independence \SC. The sorting error has a consistent effect on the tau value (Figure(a)). Imputation error has a strong negative impact at lower error rates, then the tau value reverses (Figure(b)). This can be explained as follows. Starting with two independent columns (i.e., a very small $|\tau|$), imputing a constant value induces more correlation between the two columns (i.e., a larger $|\tau|$). Once many records have been imputed with a constant value, the correlation decreases until one column contains the same constant everywhere, implying zero correlation between them (i.e., $|\tau| = 0$). 

Figure~\ref{fig:imp_dep} shows the results w.r.t. a dependence \SC. As more errors were generated, $|\tau|$ decreased, so the columns were evaluated as more independent. With the error rate at 0.8, $|\tau|$ approached zero, indicating an independence relationship.

\vspace{.25em}

\noindent \textbf{\underline{Exp-3: Evaluation of Error-Detection Approaches.}} We compared CODED with existing error-detection methods on the Boston dataset. For dependence \SCs, we compared with DC and DBoost; for independence \SCs, since DC cannot express independence relationships (cf. Section~\ref{sec:scvsic}), we compared only with DBoost. We considered both marginal and conditional \SCs, as summarized in Table~\ref{tbl:constraints}. 

\noindent \underline{Marginal \SCs: $N \nindep D$ and $R \indep B$.} Figure~\ref{fig:dep_bar_chart} and Figure~\ref{fig:indep_bar_chart} show the results for $N \nindep D$ and $R \indep B$, respectively. We consider three error levels, depending on the average error rate for the $N$ column: \emph{minor error} =  $1\%-20\%$, \emph{moderate error} = $20\%-45\%$, and \emph{major error} = $50\%-80\%$. We reported precision@50. 

We first examine CODED's performance.  CODED's precision increases with the error level. 
The reason is that when there is a larger portion of errors, the degree of dependence/independence changes more, thus it is easier for CODED to detect violations. 



We next compared the performance of CODED, DC and DBoost. As shown in Figure~\ref{fig:dep_bar_chart} and Figure~\ref{fig:indep_bar_chart}, CODED outperformed the other two approaches. DC did not perform well because the specified denial constraint (i.e., if $r_1[N] > r_2[N]$, then $r_1[D] < r_2[D]$) did not always hold, which led to many false positives. CODED outperformed DBoost due to two reasons. First, DBoost derived correlations from dirty data, and then leveraged the derived correlations to detect errors. However, since data is dirty, the derived correlations might be wrong. Second, DBoost is designed to detect outliers but the dataset has erroneous values (e.g., imputed mean values) that look like a normal value. Thus, DBoost failed to detect these errors. 

We compared the F-score of CODED, DC and DBoost using different $K$ values (moderate error level). 
Results are shown in Figure~\ref{fig:dep_score} and Figure~\ref{fig:indep_score} for $N \nindep D$ and $R \indep B$, respectively. We can see that CODED achieved significant higher F-score than DC and DBoost for all settings. 
CODED's performance depends on the error type: It performed better for sorting error and combination error, where the average F-score is $0.6$ and the max F-score is around $0.8$. But for imputation error, the average F-score and the max F-score decrease to $0.5$ and $0.6$, respectively. 
As we explained above, if errors have a small impact on \SCs, the power of using \SCs to detect the errors decreases.


\vspace{.25em}

\noindent \underline{Conditional \SCs:  $T \nindep B ~|~ C$ and $N \indep B~ | ~T$.} We also examined the effectiveness of CODED for {\em conditional} dependence and independence \SCs, with moderate error level.  The results are very similar to non-conditional cases (see Appendix~\ref{app:cond-exp} for more detail). 

\vspace{.25em}


\noindent\textbf{\underline{Exp-4: Effectiveness on Categorical Data.}} So far, we have only focused on numerical data using the $\tau$ test. Next, we use the $\chi^2$ test and evaluate the effectiveness of CODED on categorical data. The conclusion is that CODED outperformed DBoost in terms of F-score for both independence and dependence \SCs. Due to the space limit, please refer to the Appendix \ref{app:categorical} for more details. 



\noindent \textbf{\underline{Exp-5: Effectiveness compared with AFD.}} We compared the effectiveness of \system and AFD on the HOSP dataset. 
Figure~\ref{fig:cmp-afd}(a) shows the result for  \textsf{\small Zipcode -> City} vs. \textsf{\small Zipcode $\nindep$ City}; Figure~\ref{fig:cmp-afd}(b) shows the result for \textsf{\small Zipcode -> State} vs. \textsf{\small Zipcode $\nindep$ State}. 


We have a number of interesting observations. First, CODED and AFD got the same F-score for $K \leq 3000$. This is because both of them achieved 100\% precision and the same recall when $K \leq 3000$.
Second, CODED's F-score continued to grow when $K > 3000$ but AFD's F-score started to decrease. This is because that AFD can only detect the errors on the right-hand side (i.e., \textsf{\small City} or \textsf{\small State}) of an AFD. If an error occurs on the left-hand side (i.e.,  \textsf{\small Zipcode}), AFD cannot detect the errors.
In contrast,  CODED considers the statistical relationship between two columns, therefore it can detect the errors in both columns. 

\begin{figure*}[t]
   \centering 
   \includegraphics[width=0.85\linewidth, height=2.7cm]{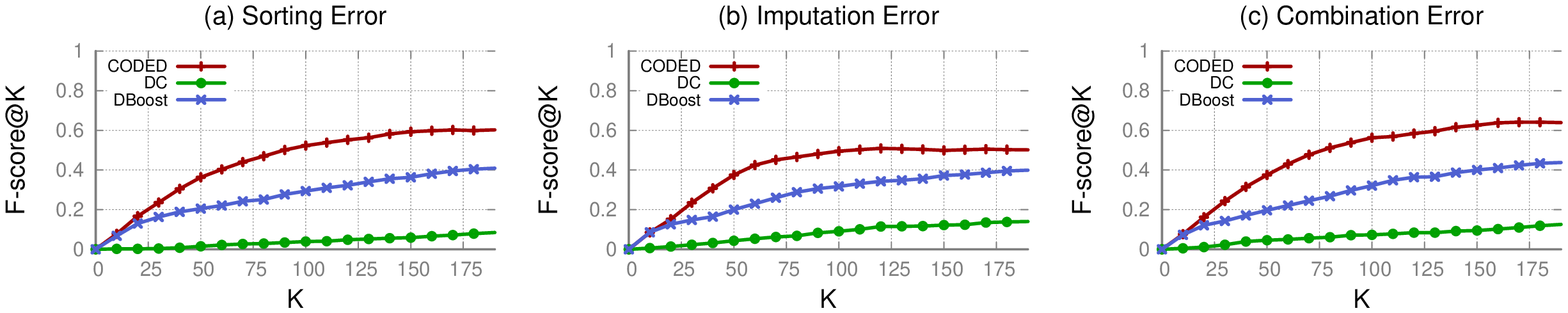} 
   \vspace{-1em}
   \caption{Effectiveness of error detection methods for dependence \SCs by varying $k$ (Boston dataset)}\label{fig:dep_score}
   \centering 
   \vspace{.25em}
   \includegraphics[width=0.85\linewidth, height=2.7cm]{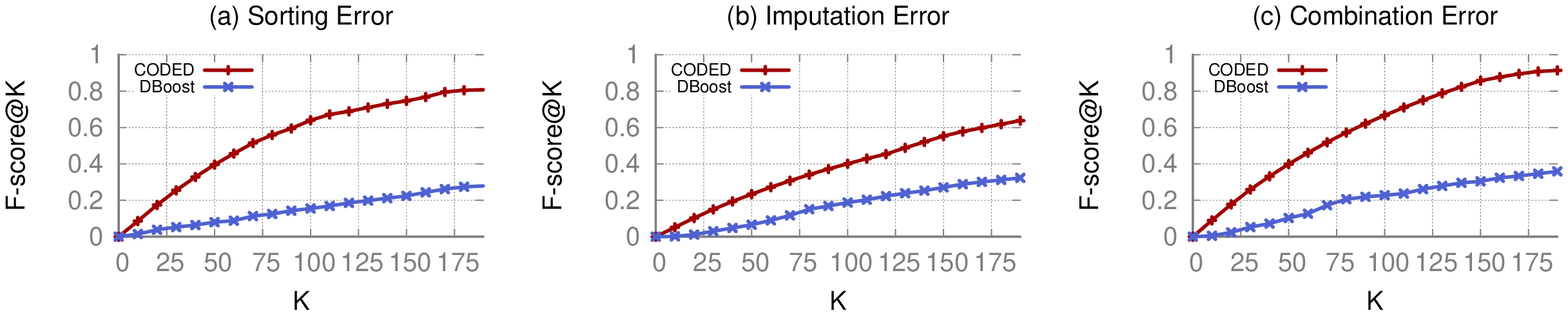} 
   \vspace{-1.5em}
   \caption{Effectiveness of error detection methods for independence \SCs by varying $k$ (Boston dataset)}\label{fig:indep_score}\vspace{-.5em}
\end{figure*}

\begin{figure*}[t]
   \centering 
   \includegraphics[width=0.85\linewidth, height=2.7cm]{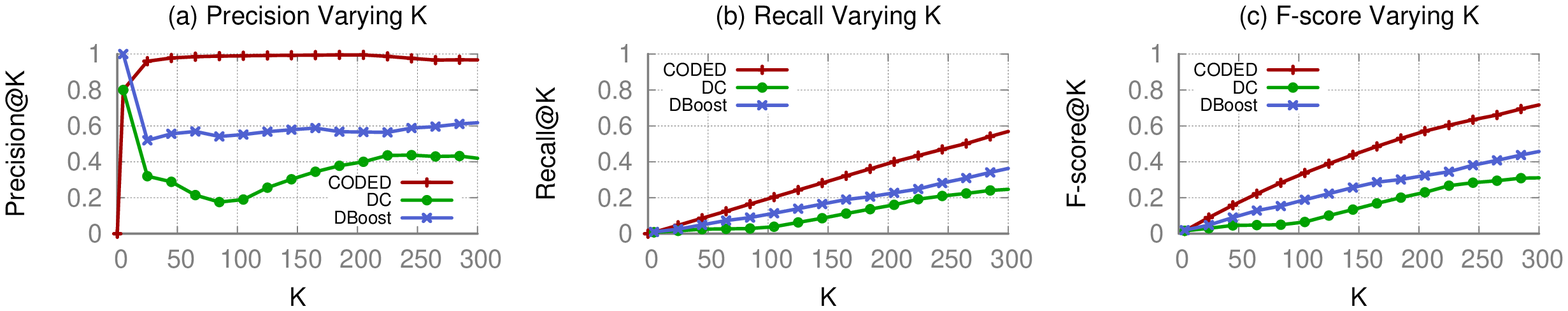}
   \vspace{-1.5em}
   \caption{Detecting real-world errors on the Sensor dataset}\label{fig:sensor_exp}
    \centering 
     \vspace{.25em}
   \includegraphics[width=0.85\linewidth, height=2.7cm]{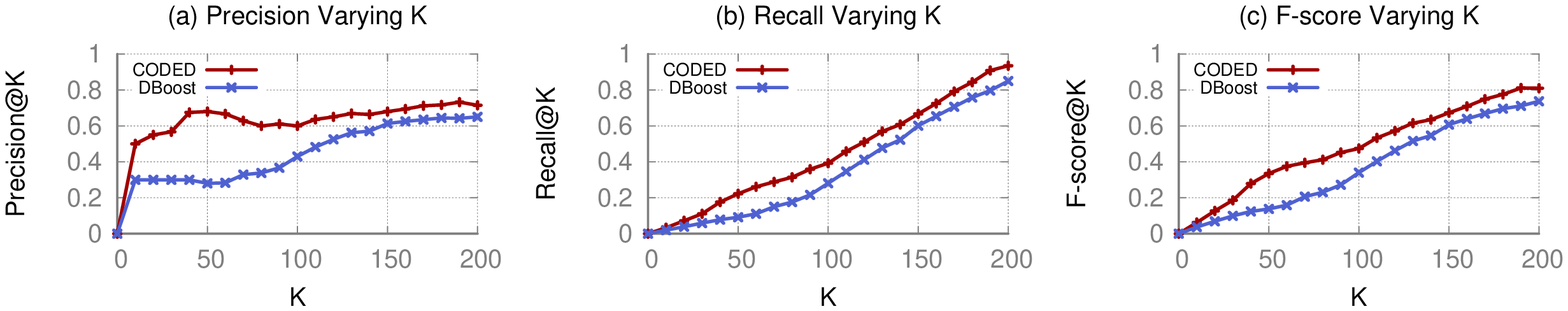}
   \vspace{-1em}
   \caption{Detecting real-world errors on the Hockey dataset}\label{fig:hockey_exp}\vspace{-1em}
\end{figure*}





\vspace{.25em}


\vspace{.25em}

\noindent \textbf{\underline{Exp-6: Effectiveness with Real-life Errors.}} We evaluated \system on the datasets with real-life errors. 


\vspace{.25em}

\noindent\underline{Sensor Dataset ($T_8 \nindep T_9$).} 
Neighboring sensors tend to report similar temperatures, which means that their readings of the temperature should be dependent. 
We specified a dependence \SC between the readings of Sensor 8 and Sensor 9.   Figure~\ref{fig:sensor_exp} shows the experimental result. The average Precision of CODED, DC, and DBoost was $0.93$, $0.65$ and $0.47$, respectively. Unlike the other two methods, the precision of CODED never decreased.
and maintained a high precision around $0.95$. For Recall and F-score, all three approaches increased continuously. CODED outperformed the other two methods with a higher increasing speed. This experiment validated the effectiveness of CODED when using it to detect real-world errors w.r.t. dependence \SCs.   




\vspace{.25em}


\noindent\underline{Hockey Dataset ($G \indep GPM \mid Y$).}  For the Hockey dataset, the columns \emph{Games(G)} and  \emph{Goal Plus-Minus(GPM)} should be independent given \emph{Draft Year(Y)}, because hockey draft studies have shown that the total number of professional games played by a player is independent of their Plus-Minus before they joined the professional league. We specified an independence \SC:  $G \indep GPM \mid Y$.  Figure~\ref{fig:hockey_exp} shows the experimental result. We only compared CODED with DBoost since DC cannot express independencies.  We can see that CODED outperformed DBoost in terms of Precision, Recall, and F-score. In particular, when k = 50, DBoost only got a precision of 0.28, but the precision of CODED was 0.68 (around 2.5 $\times$) higher.  This experiment validated the effectiveness of CODED when using it to detect real-world errors w.r.t. independence \SCs.



\vspace{.25em}

\noindent \textbf{\underline{Exp-7: Scalability.}} Our evaluation provided evidence that error drill down scales well in terms of both $k$ and $n$. 
Please refer to Appendix~\ref{app:scalibility} for more details.

\vspace{.25em}

\section{Conclusion and Future Work}\label{con}

A statistical constraint (\SC) represents a probabilistic association, or its absence, among columns in a data table. \SCs provide a powerful expressive formalism for capturing a user's domain knowledge. 
This paper explored how to exploit \SCs in data cleaning, by identifying the data errors that lead to their violation. We 
compared 
\SCs and traditional integrity constraints, and identified types of situations where \SCs add expressive power: 
when the user wishes to assert the irrelevance of one set of columns to another, and when the user expects an inferential  relationship between columns to hold not precisely, but only approximately to a certain degree, with exceptions. Our CODED system leverages \SCs for error detection by addressing three challenges: For \SC violation detection, we 
showed how well-established statistical metrics ($\chi^2$ and Kental's $\tau$) can be used to quantify the degree to which an \SC is violated. To explain violations, 
we proposed an error-drill-down framework, and devised efficient algorithms to identify the top-$k$ records that contribute the most to the violation of an SC. For checking the consistency of a set of input \SCs, we 
described an inference-based consistency-checking algorithm, which is pseudo-polynomial in terms of the number of input \SCs and the largest number of variables that occurs in any of the input \SCs. We conducted extensive experiments on real-world datasets with both synthetic and real-life errors, as well as a range of constraint types. The results showed that SCs were effective in detecting data errors that violate them, compared to state-of-the-art approaches. 

CODED represents a novel approach to leverage powerful statistical methods for error detection. It has great potential to be useful in practice, and opens a new set of research directions in the intersection of statistics and data management. In the future, we plan to extend CODED from two aspects. (1) {\em Human-in-the-Loop.} Integrate the discovery and validation of \SCs to help the user to discover and validate them 
efficiently. (2) {\em Data Repairing.} Extend CODED  to the error-repairing stage, to automatically repair errors so that the cleaned data satisfies a set of given \SCs.   

\linespread{1}
\begin{appendix}\label{sec:appendix}

\section{Efficient Implementation for $K$ Strategy}\label{app:efficient-k}

We discuss the efficient Implementation of the $K$ strategy detailed here. As illustrated in Section~\ref{subsec:top-k-algo}, we choose another data structure \emph{segment tree} in the initialization phase so as to greatly reduce the time complexity. Before applying this idea, some pre-processing is needed. We first sort dataset $D=<x_1,y_1>,...<x_n,y_n>$ by $X$ column. Then we scan the new data set to obtain the concordant pairs, and non-concordant pairs. Later, we  insert a segement tree $[y_i,y_i]$ to the segment tree. After the initialization phase, we continue do the iterations as shown in the Algorithm~\ref{alg:eff_tau}. Algorithm~\ref{alg:eff_tau} illustrates the pseudo-code. This efficient implementation is also applicable to $K^c$ strategy.

\section{Scalability.}\label{app:scalibility}
We replicated the Boston dataset to enlarge its data size, and chose a dependence: \SC $N \nindep D$. We examined the execution time of CODED by varying $k$ and $n$ (\# Records), respectively. The results are shown in Figure~\ref{fig:scalability}. 

Recall that the time complexity of CODED (the $k$ strategy) is $\mathcal{O}(n \log n)$ for initialization, and is $\mathcal{O}(kn\log n)$ for selecting $k$ records. The results are consistent with the complexity analysis, and demonstrate the good scalability of CODED w.r.t. $k$ and $n$. 
\begin{figure}[h]
   \centering 
   \includegraphics[width=1\linewidth]{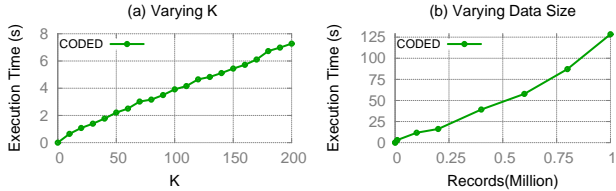} 
   \vspace{-2em}
   \caption{Scalability of CODED (Boston dataset)
   }\label{fig:scalability}
\end{figure}

\begin{algorithm}[h]
    \caption{\small Efficient $\tau$-test-based error detection algorithm}\label{alg:eff_tau}
    \KwIn{An $\SC=X\indep(\nindep) Y$, Dataset $D=\{<x_1,y_1>,...<x_n,y_n>\}$, $k$}
    \KwOut{$k$ records}
    $T \leftarrow \emptyset$; //{\small  Segment Tree} \\
    $Q \leftarrow \emptyset$; // {\small Priority Queue} \\
    $R \leftarrow \emptyset$; // {\small  Returned List} \\
    $\textsf{benefit}(<x_i,y_i>) = 0$ for $<x_i,y_i> \in D$ \;
    Sort $D$ by $X$ column value \;
    // Initialization \\
    \For{$<x_i,y_i> \in D$} 
    {
    $n_c = T.query([-\infty,y_i])$\;
    $n_d = T.query([y_i,+\infty])$\;
    $\textsf{benefit}(<x_i,y_i>) = 2n_c + (|D|-n_c-n_d)$\;
    $T.insert([y_i,y_i])$\;
    $Q.push(<x_i,y_i>,\textsf{benefit}(<x_i,y_i>)$\;
    }
    
    // Iteration \\
    \For{$i = 1$ to $k$}{
        Add $Q.top()$ to $R$\;
        Update $Q;$ // {\small Update the weight of each record $\textsf{benefit}(<x_i,y_i>) \in Q$ by querying the segment tree $T$} \\
    
    }
   {\bf return} R\;
\end{algorithm}

\begin{figure*}[t]
   \centering 
   \includegraphics[width=0.9\linewidth]{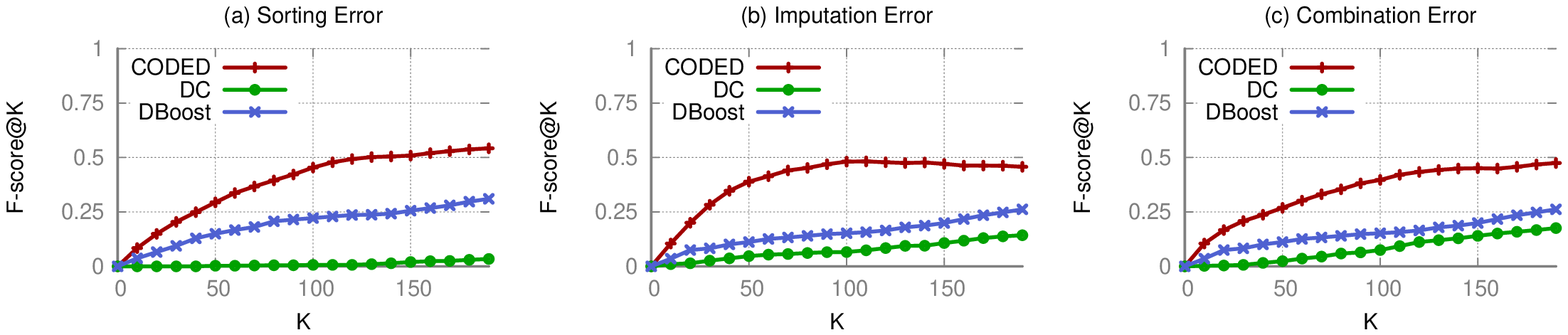} \vspace{-1em}
   \caption{Effectiveness of error detection methods for conditional dependence \SCs (Boston dataset, $K=50$)}\label{fig:dep_score_con}
\vspace{-.5em}
   \centering 
   \includegraphics[width=0.9\linewidth]{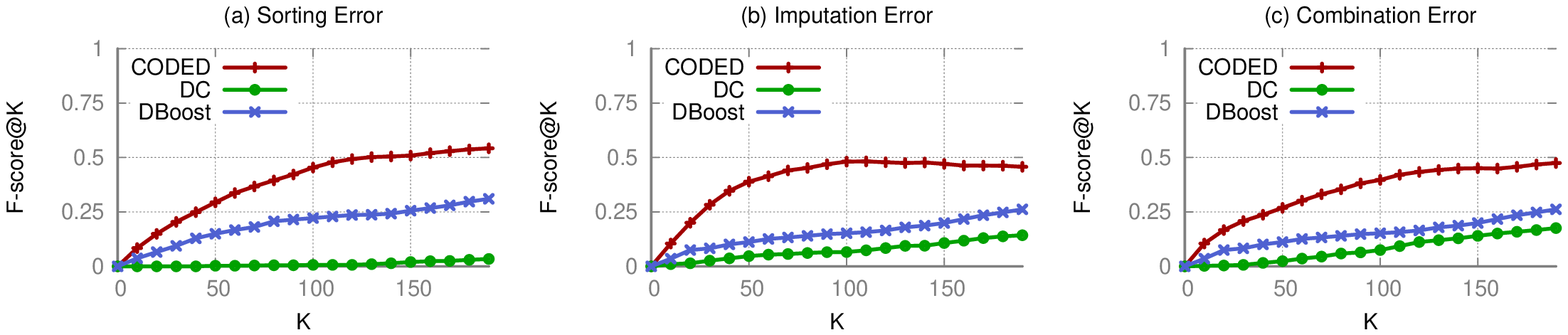}\vspace{-1em}
   \caption{Effectiveness of error detection methods for conditional independence \SCs (Boston dataset, $K=50$)}
   \label{fig:indep_score_con}
\end{figure*}

\section{Effectiveness on Conditional Independence \SCs}\label{app:cond-exp}
\noindent \underline{$T \nindep B ~|~ C$ and $N \indep B~ | ~T$.} We also examined the effectiveness of CODED for {\em conditional} dependence and independence \SCs, with moderate error level. Constraints used for CODED, DC and DBoost are summarized in Table~\ref{tab:con-summaries}. To accommodate the variables of constraints to CODED framework, we further bucket the value of \emph{Tax Rate} into categorical values. We categorize the numerical value into $6$ buckets following its distribution. We reported their F-score under different $K$ values at moderate error level, which is similar to Exp-3. Results are shown in Figure~\ref{fig:dep_score_con} and Figure~\ref{fig:indep_score_con} respectively. CODED obtained a higher F-score than another two state-of-art methods, which is even more significant under independence  \SCs. Also, similar to Exp-3 and Exp-4, the detection effectiveness of imputation error is not as stable as that of another two types errors. The average F-score of CODED under independence \SCs is $0.55$ and with max F-score $0.69$, however the two results will decrease to $0.43$ and $0.51$ with imputation error inserted. This again demonstrates that the CODED can detect the power of errors as if the power is not as significant, CODED is not that powerful.


\begin{figure}[t]
  \centering 
  \includegraphics[width=1\linewidth]{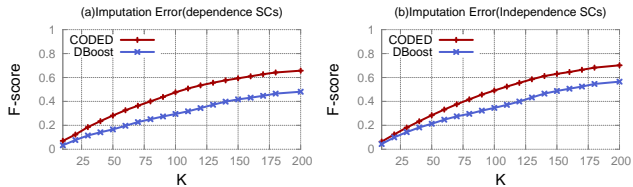} 
  \vspace{-2em}
  \caption{Effectiveness of error detection methods for categorical data (Car dataset)
  }\label{fig:score_cat}\vspace{-.5em}
\end{figure}

\section{Effectiveness on Categorical Data.}\label{app:categorical}
We use the $\chi^2$ test and evaluate the effectiveness of CODED on categorical data. We selected two \SCs ($BP\nindep Cl$ and $SA\indep DR$) on the Car dataset. DC is not applicable here because there are too many violations for the feasible DCs that we constructed.
We compared the performance of CODED and DBoost at the moderate error level.
Figure~\ref{fig:score_cat} shows the results. Due to the space limit, we just focus on the imputation errors in this experiment. The average F-score of CODED and DBoost are $0.45$ and $0.24$, respectively. Similar to Exp-3, CODED outperformed DBoost in terms of F-score for both independence and dependence \SCs.

\vspace{.25em}

\end{appendix}   

{\bibliographystyle{abbrv}
\bibliography{ref.bib}
}

\end{document}